\documentclass[11pt,a4paper]{article}
\usepackage{amssymb}
\usepackage{amsmath}
\usepackage{amsfonts}
\usepackage{bbm}
\usepackage{amsthm}
\usepackage{mathrsfs}
\usepackage{hyperref}
\usepackage{color}
\usepackage{xcolor}
\usepackage[margin=2.41cm]{geometry}
\usepackage[all,cmtip,2cell]{xy}
\UseAllTwocells
\usepackage[utf8]{inputenc}
\usepackage{graphicx}
\usepackage{varwidth}
\usepackage{comment}

\usepackage{upgreek}
\usepackage{rotating}

\usepackage{tikz}
\usetikzlibrary{shapes.geometric}

\definecolor{darkred}{rgb}{0.8,0.1,0.1}
\hypersetup{
     colorlinks=false,         
     linkcolor=darkred,
     citecolor=blue,
}

\theoremstyle{plain}
\newtheorem{theo}{Theorem}[section]
\newtheorem{lem}[theo]{Lemma}
\newtheorem{propo}[theo]{Proposition}
\newtheorem{cor}[theo]{Corollary}

\theoremstyle{definition}
\newtheorem{defi}[theo]{Definition}

\newtheorem{exx}[theo]{Example}
\newenvironment{ex}
  {\pushQED{\qed}\exx}
  {\popQED\endexx}

\newtheorem{remm}[theo]{Remark}
\newenvironment{rem}
  {\pushQED{\qed}\remm}
  {\popQED\endremm}

\numberwithin{equation}{section}

\def\bbK{\mathbb{K}}
\def\bbR{\mathbb{R}}
\def\bbC{\mathbb{C}}

\def\bbL{\mathbb{L}}

\def\Hom{\mathrm{Hom}}

\def\id{\mathrm{id}}

\def\dd{\mathrm{d}}

\def\cc{\mathrm{c}}

\def\1{I}
\def\oone{\mathbbm{1}}
\def\op{\mathrm{op}}

\def\Diag{\mathrm{Diag}}

\def\Loc{\mathbf{Loc}}
\def\Lan{\operatorname{Lan}}

\def\Locc{\mathbf{Loc}_{\text{\large $\diamond$}}^{}}
\def\Man{\mathbf{Man}}
\def\Manc{\mathbf{Disk}}

\def\Set{\mathbf{Set}}
\def\Alg{\mathbf{Alg}}
\def\Vec{\mathbf{Vec}}
\def\Ch{\mathbf{Ch}}

\def\FT{\mathbf{FT}}

\def\CC{\mathbf{C}}
\def\DD{\mathbf{D}}

\def\EE{\mathbf{E}}

\def\MM{\mathbf{M}}

\def\OCat{\mathbf{OrthCat}}

\def\Seq{\mathbf{Seq}}
\def\Op{\mathbf{Op}}

\def\QFT{\mathbf{QFT}}
\def\LFT{\mathbf{LFT}}

\def\AAA{\mathfrak{A}}
\def\LLL{\mathfrak{L}}
\def\BBB{\mathfrak{B}}

\def\CCC{\mathfrak{C}}
\def\DDD{\mathfrak{D}}

\def\O{\mathcal{O}}
\def\P{\mathcal{P}}

\def\Q{\mathcal{Q}}

\def\Qlin{\mathfrak{Q}_\mathsf{lin}^{~}}
\def\Ulin{\mathfrak{U}_\mathsf{lin}^{~}}

\def\colim{\mathrm{colim}}

\newcommand\und[1]{\underline{#1}}
\newcommand\ovr[1]{\overline{#1}}

\DeclareMathOperator*{\Motimes}{\text{\raisebox{0.25ex}{\scalebox{0.8}{$\bigotimes$}}}}

\def\sk{\vspace{2mm}}

\makeatletter
\let\@fnsymbol\@alph
\makeatother

\definecolor{darkgreen}{rgb}{0,0.6,0}

\title{%
Algebraic field theory operads and linear quantization
}

\author{%
Simen Bruinsma$^{a}$\ and\ Alexander Schenkel$^{b}$ \vspace{4mm}\\
{\small School of Mathematical Sciences, University of Nottingham,}\\
{\small University Park, Nottingham NG7 2RD, United Kingdom.}\vspace{4mm}\\
{\small \begin{tabular}{ll}
Email: & ${}^a$~\texttt{simen.bruinsma@nottingham.ac.uk}\\
& ${}^b$~\texttt{alexander.schenkel@nottingham.ac.uk}\vspace{2mm}
\end{tabular}
}
}

\date{September 2019}

\begin{document}

\maketitle

\vspace{-4.5mm}

\begin{abstract}
\noindent We generalize the operadic approach to algebraic quantum field theory [arXiv:1709.08657] to a broader class of field theories whose observables on a spacetime are algebras over any single-colored operad. A novel feature of our framework is that it gives rise to adjunctions between different types of field theories. As an interesting example, we study an adjunction whose left adjoint describes the quantization of linear field theories. We also analyze homotopical properties of the linear quantization adjunction for chain complex valued field theories, which leads to a homotopically meaningful quantization prescription for linear gauge theories.
\end{abstract}

\vspace{-1mm}

\paragraph*{Keywords:} algebraic quantum field theory, locally covariant quantum field theory, 
colored operads, universal constructions, gauge theory, model categories

\vspace{-2mm}

\paragraph*{MSC 2010:} 81Txx, 18D50, 18G55

\vspace{-1mm}

\renewcommand{\baselinestretch}{0.87}\normalsize
\tableofcontents
\renewcommand{\baselinestretch}{1.0}\normalsize




\section{\label{sec:intro}Introduction and summary} 
Algebraic quantum field theory is a well-established  and successful
framework to axiomatize and investigate quantum field theories
on the Minkowski spacetime and also on more general Lorentzian manifolds,
see e.g.\ \cite{HaagKastler,AQFTbook} for overviews. In this setting a theory
is described by a functor $\AAA : \CC\to\Alg_{\mathsf{As}}$ from a suitable
category $\CC$ of spacetimes to the category of associative and unital algebras,
which is required to satisfy some physically motivated axioms. For instance,
in locally covariant algebraic quantum field theory \cite{Brunetti,FewsterVerch},
$\CC=\Loc$ is the category of globally hyperbolic Lorentzian manifolds
with morphisms given by causal isometric embeddings and the physical axioms are Einstein causality
and the time-slice axiom. Einstein causality demands that any two observables, i.e.\ elements
of the algebras assigned by $\AAA$, that are causally disjoint
commute with each other, which encodes the idea that two measurements
in causally disjoint spacetime regions do not influence each other.
The time-slice axiom demands that the algebra maps $\AAA(f): \AAA(M)\to\AAA(M^\prime)$
associated to Cauchy morphisms, i.e.\ spacetime embeddings $f:M\to M^\prime$ such that 
$f(M)\subseteq M^\prime$ contains a Cauchy surface of $M^\prime$, are isomorphisms,
which encodes a concept of time evolution. The framework of algebraic quantum field
theory can also be adapted to obtain a novel point of view on classical field theories, see e.g.\
\cite{BrunettiFredenhagenRibeiro,FredenhagenRejzner,Collini,BeniniSchenkelPoisson},
where in contrast to associative and unital algebras one assigns Poisson algebras of classical observables
to spacetimes. The classical analog of Einstein causality then demands that the Poisson bracket between 
any two causally disjoint observables is zero.
\sk

The aim of this paper is to develop an operadic framework that generalizes \cite{BeniniSchenkelWoike}
to a very broad and flexible class of field theories, see Definition \ref{def:fieldtheory}. 
This includes as special instances the various flavors of algebraic quantum field theory
\cite{HaagKastler,Brunetti,AQFTbook,FewsterVerch} and their classical analogs 
\cite{BrunettiFredenhagenRibeiro,FredenhagenRejzner,Collini,BeniniSchenkelPoisson}.
Our two main motivations for this work are as follows: (1)~Describing field theories
in terms of algebras over colored operads provides an excellent framework
to discover and study universal constructions. This has already lead to a 
refinement of Fredenhagen's universal algebra construction for quantum field theories
\cite{Fre1,Fre2,Fre3,Lang} in terms of a so-called operadic left Kan extension
\cite{BeniniSchenkelWoike}, which technically behaves better than the original construction
as it respects the quantum field theory axioms. 
In this paper we will show that the quantization of linear field theories may be expressed in terms of an operadic left Kan extension too, 
which allows us to understand and describe the interplay between quantization and other universal constructions.
(2)~Operadic techniques are particularly useful and powerful when working with chain complex 
valued field theories, e.g.\ gauge theories described in the BRST/BV formalism
\cite{Hollands,FredenhagenRejzner,FredenhagenRejzner2}. The reason for this
is that chain complexes are naturally compared by quasi-isomorphisms and hence one is
only allowed to perform constructions that  preserve quasi-isomorphisms.
Operad theory provides a huge toolbox to develop such constructions,
which in technical language are called derived functors, see e.g.\ \cite{HinichOriginal,Hinich}
and also \cite{BeniniSchenkelWoikehomotopy} for applications to quantum field theory.
In this paper we apply these techniques to investigate homotopical
properties of the linear quantization functor. A similar construction in the context
of factorization algebras \cite{CostelloGwilliam} has been 
recently investigated in \cite{GwilliamHaugseng}. As a simple example, we
present a quantization of linear Chern-Simons theory on oriented surfaces 
that is compatible with quasi-isomorphisms.
\sk

Let us now present a more detailed outline of the content of this paper.
In Section \ref{sec:preliminaries} we shall fix our notations by recalling
the necessary background material on colored operads and their algebras.
In Section \ref{sec:FTs} we introduce our broad and flexible framework
for field theories in Definition \ref{def:fieldtheory}. 
A field theory is described by a
functor $\AAA : \CC\to \Alg_\P$ from a small category
$\CC$ to the category of algebras over a single-colored
operad $\P$, which is required to satisfy a suitable generalization of the Einstein causality axiom.
(The time-slice axiom will be formalized via localization techniques in Section \ref{subsec:orthlocalication}.)
One should interpret $\CC$ as a category of spacetimes and $\P$ as the operad controlling 
the algebraic structure of the observables on a fixed spacetime.
For example, quantum field theories are obtained by choosing the associative operad $\P=\mathsf{As}$
and classical field theories by choosing the Poisson operad $\P = \mathsf{Pois}$.
Linear field theories, which we describe in terms of Heisenberg Lie algebras of 
presymplectic vector spaces, are obtained by choosing the unital Lie operad $\P= \mathsf{uLie}$.
One of the key results of this section (see Theorem \ref{theo:FTcatiso}) 
is that such field theories are precisely the algebras over a colored operad
that we denote by $\P^{(r_1,r_2)}_{\ovr{\CC}}$. This colored operad 
depends on two different kinds of input data, which control the spacetime category 
of interest and the type of field theory.
More precisely, the first datum is an orthogonal category $\ovr{\CC} = (\CC,\perp)$ (see Definition
\ref{def:orthogonalcat}) and the second is a bipointed single-colored
operad $\P^{(r_1,r_2)} = (\P , r_1,r_2 : I[2]\rightrightarrows \P)$ (see Definition \ref{def:bipointed}). 
The orthogonality relation $\perp$ and the two pointings $r_1,r_2$ are required to formalize
a suitable generalization of the Einstein causality axiom.
We prove that the assignment $(\ovr{\CC},\P^{(r_1,r_2)}) \mapsto  \P^{(r_1,r_2)}_{\ovr{\CC}}$
of field theory operads is in a suitable sense functorial.
\sk

In Section \ref{sec:adjunctions} we harness this functorial behavior in order to study
adjunctions between the categories of field theories
corresponding to different $\ovr{\CC}$ and  $\P^{(r_1,r_2)}$.
This includes generalizations of the time-slicification and local-to-global adjunctions 
from \cite{BeniniSchenkelWoike}, which have already found interesting applications 
to quantum field theory on spacetimes with boundaries \cite{BeniniDappiaggiSchenkel}.
A novel feature of our framework, which is not captured by 
\cite{BeniniSchenkelWoike}, is a second kind of functorial assignment
$\P^{(r_1,r_2)} \mapsto  \P^{(r_1,r_2)}_{\ovr{\CC}}$ of our colored operads
to bipointed single-colored operads. This results in adjunctions
between the categories of field theories of different types.
We shall investigate in detail the interplay of such adjunctions with
the time-slice axiom and local-to-global property of field theories.
A particularly interesting example, which we study in detail in Section \ref{sec:quantization}, 
is given by an adjunction whose left adjoint describes the quantization
of linear field theories. 
\sk

In Section \ref{sec:outlook} we extend our results
to the case of $\Ch(\bbK)$-valued field theories, i.e.\ gauge theories, 
by using techniques from model category theory \cite{Dwyer,Hovey}.
Our reformulation in Section \ref{sec:quantization} of the usual
quantization of linear field theories in terms of (the left adjoint of) 
an adjunction is very valuable for studying the quantization of linear gauge theories. 
In particular, it allows us to construct a {\em derived} linear quantization functor 
which provides a homotopically meaningful  quantization prescription for linear 
gauge theories in the sense that it maps weakly  equivalent linear gauge theories 
to weakly equivalent quantum gauge theories. A deeper homotopical study of the building blocks of 
the derived linear quantization functor (see Appendix \ref{app:envelop}) 
reveals that it can be modeled (up to a natural weak equivalence)
by the underived linear quantization functor. From a computational point of view,
this is a very pleasing result because it allows us to write down explicit formulas for the quantization
of linear gauge theories. This will be illustrated by studying linear Chern-Simons theory 
on oriented surfaces. We conclude by analyzing in some detail the interplay between our (derived) 
linear quantization functor and suitable homotopical generalizations of the time-slice axiom 
and local-to-global property of gauge theories.


\section{\label{sec:preliminaries}Preliminaries}
Throughout this paper we fix a closed symmetric monoidal category $\MM$,
which we further assume to be complete and cocomplete, i.e.\ all small limits and colimits exist in $\MM$.
The monoidal product is denoted by $\otimes : \MM\times \MM\to \MM$, the monoidal unit by $I\in \MM$ 
and the internal hom by $[-,-]: \MM^\op \times \MM\to \MM$, where $(-)^\op$ denotes
the opposite category. The symmetric braiding is denoted by $\tau : m\otimes m^\prime \to m^\prime\otimes m$,
for all $m,m^\prime\in\MM$.
We shall always suppress the associator and unitors
and in particular simply write $m_1\otimes \cdots \otimes m_n$ for multiple tensor products
of objects $m_1,\dots,m_n\in \MM$. Because $\MM$ is by assumption cocomplete,
there exists a $\Set$-tensoring $\otimes :\Set \times \MM\to\MM$, which
we denote with abuse of notation by the same symbol as the monoidal product.
Explicitly, for any set $S\in\Set$ and $m\in\MM$, we define
\begin{flalign}\label{eqn:Settensoring}
S\otimes m \,:=\,\coprod_{s\in S} m\in\MM\quad,
\end{flalign}
where $\coprod$ is the coproduct in $\MM$.
\begin{ex}\label{ex:Set}
A simple example of a bicomplete closed symmetric monoidal category
is the Cartesian closed category $\Set$ of sets. Here $\otimes = \times $ is the Cartesian product,
$I =\{\ast\}$ is any singleton set and $[S,T] = \mathrm{Map}(S,T)$ is the set of maps from $S$ to $T$.
The symmetric braiding is given by the flip map $\tau : S\times T\to T\times S\,,~(s,t)\mapsto (t,s)$.
\end{ex}

\begin{ex}\label{ex:Vec}
Another standard example of a bicomplete closed symmetric monoidal
category is the category $\Vec_\bbK$ of vector spaces over a field $\bbK$.
Here $\otimes$ is the usual tensor product of vector spaces, $I=\bbK$ is the $1$-dimensional
vector space and $[V,W] = \mathrm{Hom}_\bbK(V,W)$ is the vector space of linear maps from $V$ to $W$.
The symmetric braiding is given by the flip map $\tau : V\otimes W \to W\otimes V\,,~ v\otimes w\mapsto w\otimes v$.
\end{ex}

\subsection{Colored operads}
We provide a brief review of those aspects of
the theory of colored operads that are relevant for
this work. We refer to \cite{Yau}, \cite{BergerMoerdijk}
and \cite{BeniniSchenkelWoike} for a more detailed presentation.
\sk

Let $\CCC\in\Set$ be a non-empty set, which we shall call the `set of colors'.
We will use the notation $\und{c} := (c_1, \dots,c_n) \in \CCC^{n}$ for elements of the $n$-fold product set. 
\begin{defi}\label{def:operad}
A {\em $\CCC$-colored operad} $\O$ with values in $\MM$ is given by the following data:
\begin{itemize}
\item for each $n\geq 0$ and $(\und{c},t) \in\CCC^{n+1}$,
an object $\O\big(\substack{t \\\und{c}}\big)\in\MM$ (called the {\em object of operations} from $\und{c}$ to $t$);

\item for each $n\geq 0$, $(\und{c},t) \in\CCC^{n+1}$ and permutation $\sigma\in\Sigma_n$,
an $\MM$-morphism $\O(\sigma) : \O\big(\substack{t \\\und{c}}\big) \to \O\big(\substack{t\\ \und{c}\sigma}\big)$
(called the {\em permutation action}), where $\und{c} \sigma := (c_{\sigma(1)},\dots,c_{\sigma(n)})$;

\item for each $n>0$, $k_1,\dots,k_n \geq 0$, $(\und{a},t)\in\CCC^{n+1}$ and $(\und{b}_i,a_i)\in\CCC^{k_i+1}$,
for $i=1,\dots,n$, an $\MM$-morphism $\gamma : \O\big(\substack{t \\\und{a}}\big)\otimes\Motimes_{i=1}^n \O\big(\substack{a_i\\ \und{b}_i}\big)\to \O\big(\substack{t \\ \und{b}}\big)$ (called the {\em operadic composition}),
where $\und{b} := (\und{b}_1,\dots,\und{b}_n)$ is defined by concatenation;

\item for each $c\in\CCC$, an $\MM$-morphism $\oone : I\to \O\big(\substack{c\\ c}\big)$ (called the {\em operadic unit}).

\end{itemize}
This data is required to satisfy the standard permutation action, 
equivariance, associativity and unitality axioms, see e.g.\ \cite[Definition 11.2.1]{Yau}. 
A morphism $\phi : \O\to\P$ between two $\CCC$-colored operads $\O$ and $\P$
with values in $\MM$ is a family of $\MM$-morphisms
\begin{flalign}
\phi\,:\, \O\big(\substack{t \\ \und{c}}\big) ~\longrightarrow ~ \P\big(\substack{t \\ \und{c}}\big)\quad,
\end{flalign}
for all $n\geq 0$ and $(\und{c},t)\in \CCC^{n+1}$, that is compatible
with the permutation actions $\phi \, \O(\sigma) = \P(\sigma)\,\phi$,
the operadic compositions $\phi\,\gamma^\O = \gamma^\P\,(\phi\otimes\Motimes_{i=1}^n \phi)$
and the operadic units $\phi\,\oone^\O = \oone^\P$. We denote the {\em category
of $\CCC$-colored operads} with values in $\MM$ by $\Op_\CCC(\MM)$.
\end{defi}

Colored operads generalize the concept of (enriched) categories in the following sense.
In contrast to allowing only for $1$-to-$1$ operations, such as the morphisms
$\CC(c,c^\prime)$ in a category $\CC$, colored operads also describe $n$-to-$1$
operations in terms of the objects of operations $\O\big(\substack{t \\ \und{c}}\big)$.
The operadic composition generalizes the usual categorical composition to
operations of higher arity and the operadic unit is analogous to the identity morphisms
in a category. Permutation actions are a new feature for operations of arity $\geq 2$
and they have no analog in ordinary category theory. 
The following example clarifies how every category
defines a colored operad with only $1$-ary operations.
\begin{ex}\label{ex:Diag}
Let $\CC$ be a small category and denote its set of objects by $\CC_0$.
The following construction defines a $\CC_0$-colored operad
$\Diag_\CC\in\Op_{\CC_0}(\Set)$ with values in $\MM=\Set$, 
which is called the {\em diagram operad} of $\CC$,
see e.g.\ \cite{BergerMoerdijk}. For $(\und{c},t)\in\CC_0^{n+1}$, one defines the set of operations by
\begin{flalign}
\Diag_{\CC}\big(\substack{t \\ \und{c}}\big) \,:=\, \begin{cases}
\emptyset & ~,~~\text{for }n\neq 1~,\\
\CC(c,t) &~,~~\text{for }\und{c} = c~.
\end{cases}
\end{flalign}
The permutation action is uniquely fixed because $\Sigma_1 = \{e\}$ is the trivial group.
The only non-trivial operadic compositions are
$\gamma : \Diag_{\CC}\big(\substack{t \\ a}\big) \times \Diag_{\CC}\big(\substack{a \\ b}\big) 
\to \Diag_{\CC}\big(\substack{t \\ b}\big)$ and they are given by composition of morphisms
in the category $\CC$. Finally, the operadic unit $\oone : \{\ast\} \to \Diag_{\CC}\big(\substack{ c \\ c}\big)$
is given by the identity morphisms in the category $\CC$.
One confirms that this defines a colored operad in the sense of Definition \ref{def:operad}.
\end{ex}

Many interesting examples of (colored) operads can be conveniently defined
in terms of generators and relations, see e.g.\ the examples below.
Let us briefly explain how this construction works.
We denote by $\Seq_\CCC(\MM)$ the category
of {\em $\CCC$-colored (non-symmetric) sequences} with values in $\MM$.
An object $X\in\Seq_\CCC(\MM)$ is a family
of objects $X\big(\substack{t \\ \und{c} }\big)\in\MM$, for 
all $n\geq 0$ and $(\und{c},t) \in\CCC^{n+1}$,
and a $\Seq_\CCC(\MM)$-morphism $f : X\to Y$
is a family of $\MM$-morphisms 
$f: X\big(\substack{t \\ \und{c} }\big)\to Y\big(\substack{t \\ \und{c} }\big)$,
for all $n\geq 0$ and $(\und{c},t) \in\CCC^{n+1}$.
There exists a forgetful functor
$U : \Op_\CCC(\MM)\to \Seq_\CCC(\MM)$ that forgets
the permutation action, operadic composition and operadic unit of
a $\CCC$-colored operad. This functor has a left adjoint
which is called the {\em free $\CCC$-colored operad functor}, i.e.\
we have an adjunction
\begin{flalign}\label{eqn:freeforgetOp}
\xymatrix{
F \,:\, \Seq_\CCC(\MM) ~\ar@<0.5ex>[r]&\ar@<0.5ex>[l]  ~\Op_\CCC(\MM) \,:\, U
}\quad.
\end{flalign}
Given any choice of generators $G\in \Seq_\CCC(\MM)$, 
we consider the corresponding free $\CCC$-colored operad $F(G)\in\Op_\CCC(\MM)$.
In order to implement relations, we consider  $R\in \Seq_\CCC(\MM)$
together with two parallel $\Seq_\CCC(\MM)$-morphisms
$r_1,r_2 : R \rightrightarrows U F (G)$. Note that because \eqref{eqn:freeforgetOp}
is an adjunction, the latter is equivalent to two parallel $\Op_\CCC(\MM)$-morphisms
$r_1,r_2 : F(R) \rightrightarrows F(G)$, which we denote with abuse of notation by the same symbols.
Because the category $\Op_\CCC(\MM)$ is cocomplete,  the following construction defines 
a $\CCC$-colored operad.
\begin{defi}\label{def:genreloperads}
The $\CCC$-colored operad presented by
the generators $G\in \Seq_\CCC(\MM)$ and relations 
$r_1,r_2 : R \rightrightarrows UF(G)$
is defined as the coequalizer
\begin{flalign}
\xymatrix@C=2.5em{ 
F(R) \ar@<-0.5ex>[r]_-{r_2}\ar@<0.5ex>[r]^-{r_1}  ~&~F(G) \ar@{-->}[r]  ~&~  F(G)/(r_1=r_2) 
}
\end{flalign}
in $\Op_\CCC(\MM)$.
\end{defi}

\begin{ex}\label{ex:Ass}
Consider for the moment $\MM = \Set$. 
The {\em associative operad} $\mathsf{As}\in \Op_{\{\ast\}}(\Set)$ is the single-colored operad 
(i.e.\ $\CCC = \{\ast\}$ is a singleton) presented by the following generators and relations: We define
the set of generators of arity $n$ by
\begin{flalign}
G(n) \, :=\, \begin{cases}
\{\eta\} & ~,~~\text{for } n=0~,\\
\{\mu\} &~,~~\text{for } n=2~,\\
\emptyset &~,~~\text{else },
\end{cases}
\end{flalign}
for all $n\geq 0$. The generator $\mu$ in arity $2$ is interpreted as a multiplication
operation and the generator $\eta$ in arity $0$ as a unit element. To implement
associativity and left/right unitality of these operations, we consider
\begin{flalign}
R(n)\,:= \, \begin{cases}
\{\lambda,\rho\}  & ~,~~\text{for } n=1~,\\
\{a \} &~,~~\text{for } n=3~,\\
\emptyset &~,~~\text{else },
\end{cases}
\end{flalign}
for all $n\geq 0$, together with the two $\Seq_{\{\ast\}}(\Set)$-morphisms
$r_1,r_2 : R\to UF(G)$ defined by
\begin{flalign}\label{tmp:assrelations}
r_1 \,:\, \begin{cases}
\lambda \, \longmapsto\, \gamma\big(\mu\otimes (\eta \otimes\oone)\big)~~,\\
\rho \, \longmapsto\, \gamma\big(\mu \otimes (\oone\otimes\eta)\big)~~,\\
a\, \longmapsto\, \gamma\big(\mu \otimes (\mu\otimes \oone)\big)~~,
\end{cases}
\quad,\qquad
r_2 \,:\, \begin{cases}
\lambda \, \longmapsto\, \oone~~,\\
\rho \, \longmapsto\, \oone~~,\\
a\, \longmapsto\, \gamma\big(\mu \otimes (\oone\otimes \mu)\big) ~~,
\end{cases}
\end{flalign}
where the operadic composition and unit
are those of the free operad $F(G)$.
The associative operad $\mathsf{As}:= F(G)/(r_1=r_2)\in \Op_{\{\ast\}}(\Set)$ 
is defined as the corresponding coequalizer. \sk

It is instructive and useful to visualize the generators and relations in terms of rooted trees.
The generators are depicted by
\begin{subequations}\label{eqn:assrelationpics}
\begin{flalign}
\mu = \parbox{1cm}{\begin{tikzpicture}
\draw (0,0.5) -- (0,0);
\draw (0,0) -- (-0.4,-0.4);
\draw (0,0) -- (0.4,-0.4);
\end{tikzpicture}}
\quad,\qquad \eta = \parbox{0.5cm}{\begin{tikzpicture}
\draw (0,0.5) -- (0,0);
\draw[fill=white] (0,0) circle (0.5mm);
\end{tikzpicture}}\quad,
\end{flalign}
and the relations (in the order they appear in \eqref{tmp:assrelations}) then read as
\begin{flalign}
\parbox{1cm}{\begin{tikzpicture}
\draw (0,0.5) -- (0,0);
\draw (0,0) -- (-0.4,-0.4);
\draw (0,0) -- (0.4,-0.4);
\draw[fill=white] (-0.4,-0.4) circle (0.5mm);
\end{tikzpicture}}
=~~ \parbox{0.5cm}{\begin{tikzpicture}
\draw (0,0.5) -- (0,-0.4) node[midway,right] {$\oone$};
\end{tikzpicture}}\quad,\qquad
\parbox{1cm}{\begin{tikzpicture}
\draw (0,0.5) -- (0,0);
\draw (0,0) -- (-0.4,-0.4);
\draw (0,0) -- (0.4,-0.4);
\draw[fill=white] (0.4,-0.4) circle (0.5mm);
\end{tikzpicture}}
=~~ \parbox{0.5cm}{\begin{tikzpicture}
\draw (0,0.5) -- (0,-0.4) node[midway,right] {$\oone$};
\end{tikzpicture}}\quad,\qquad
\parbox{1cm}{\begin{tikzpicture}
\draw (0,0.5) -- (0,0);
\draw (0,0) -- (-0.2,-0.2);
\draw (0,0) -- (0.4,-0.4)  ;
\draw (-0.2,-0.2) -- (-0.4,-0.4) ;
\draw (-0.2,-0.2) -- (0,-0.4) ;
\end{tikzpicture}}
=~~
\parbox{1cm}{\begin{tikzpicture}
\draw (0,0.5) -- (0,0);
\draw (0,0) -- (-0.4,-0.4);
\draw (0,0) -- (0.2,-0.2);
\draw (0.2,-0.2) -- (0,-0.4);
\draw (0.2,-0.2) -- (0.4,-0.4);
\end{tikzpicture}}\quad.
\end{flalign}
\end{subequations}

Let us note that the associative operad can be defined in any bicomplete closed symmetric
monoidal category $\MM$. Using the $\Set$-tensoring \eqref{eqn:Settensoring}
and the unit object $I\in\MM$, we define generators $G\otimes I\in\Seq_{\{\ast\}}(\MM)$
and relations $r_1\otimes I, r_2\otimes I : R\otimes I \to UF(G)\otimes I \cong UF(G\otimes I)$ 
in the category of $\MM$-valued sequences $\Seq_{\{\ast\}}(\MM)$. 
The corresponding coequalizer then defines the $\MM$-valued associative operad  $\mathsf{As}:= 
F(G\otimes I)/ ( r_1\otimes I=r_2\otimes I)\in \Op_{\{\ast\}}(\MM)$.
\end{ex}

\begin{ex}\label{ex:Lie}
Consider for the moment $\MM = \Vec_\bbK$. The {\em Lie operad}
$\mathsf{Lie}\in\Op_{\{\ast\}}(\Vec_\bbK)$ is the single-colored operad
presented by the following generators and relations: There is only one 
generator of arity $2$, the Lie bracket, that we depict by
\begin{subequations}\label{eqn:Lierelationpics}
\begin{flalign}
[\cdot,\cdot] = \parbox{1cm}{\begin{tikzpicture}
\draw (0,0.5) -- (0,0);
\draw (0,0) -- (-0.4,-0.4);
\draw (0,0) -- (0.4,-0.4);
\draw[fill=black] (0,0) circle (0.6mm);
\end{tikzpicture}}\quad.
\end{flalign}
The relations are given by antisymmetry and the Jacobi identity
\begin{flalign}
\parbox{1cm}{\begin{tikzpicture}[baseline=21]
\draw (0,0.5) -- (0,0);
\draw (0,0) -- (-0.4,-0.4) node [below=-1pt] {\tiny{1}};
\draw (0,0) -- (0.4,-0.4) node [below=-1pt] {\tiny{2}};
\draw[fill=black] (0,0) circle (0.6mm);
\end{tikzpicture}} = ~-\!\!\!\parbox{1cm}{\begin{tikzpicture}[baseline=21]
\draw (0,0.5) -- (0,0);
\draw (0,0) -- (-0.4,-0.4) node [below=-1pt] {\tiny{2}};
\draw (0,0) -- (0.4,-0.4) node [below=-1pt] {\tiny{1}};
\draw[fill=black] (0,0) circle (0.6mm);
\end{tikzpicture}}\quad,\qquad
\parbox{1cm}{\begin{tikzpicture}[baseline=21]
\draw (0,0.5) -- (0,0);
\draw (0,0) -- (-0.4,-0.4) node [below=-1pt] {\tiny{1}};
\draw (0,0) -- (0.2,-0.2);
\draw (0.2,-0.2) -- (0,-0.4) node [below=-1pt] {\tiny{2}};
\draw (0.2,-0.2) -- (0.4,-0.4) node [below=-1pt] {\tiny{3}};
\draw[fill=black] (0,0) circle (0.6mm);
\draw[fill=black] (0.2,-0.2) circle (0.6mm);
\end{tikzpicture}} ~+ \parbox{1cm}{\begin{tikzpicture}[baseline=21]
\draw (0,0.5) -- (0,0);
\draw (0,0) -- (-0.4,-0.4) node [below=-1pt] {\tiny{2}};
\draw (0,0) -- (0.2,-0.2);
\draw (0.2,-0.2) -- (0,-0.4) node [below=-1pt] {\tiny{3}};
\draw (0.2,-0.2) -- (0.4,-0.4) node [below=-1pt] {\tiny{1}};
\draw[fill=black] (0,0) circle (0.6mm);
\draw[fill=black] (0.2,-0.2) circle (0.6mm);
\end{tikzpicture}} ~+ \parbox{1cm}{\begin{tikzpicture}[baseline=21]
\draw (0,0.5) -- (0,0);
\draw (0,0) -- (-0.4,-0.4)node [below=-1pt] {\tiny{3}};
\draw (0,0) -- (0.2,-0.2);
\draw (0.2,-0.2) -- (0,-0.4) node [below=-1pt] {\tiny{1}};
\draw (0.2,-0.2) -- (0.4,-0.4) node [below=-1pt] {\tiny{2}};
\draw[fill=black] (0,0) circle (0.6mm);
\draw[fill=black] (0.2,-0.2) circle (0.6mm);
\end{tikzpicture}} ~=~0\quad,
\end{flalign}
\end{subequations}
where the numbers below the trees indicate input permutations.
\sk

Note that for defining the Lie relations we had to use
the natural Abelian group structure on the $\Hom$-sets of $\Vec_\bbK$, 
i.e.\ addition of linear maps between vector spaces. Hence,
the Lie operad can {\em not} be defined in a generic 
bicomplete closed symmetric monoidal category $\MM$.
If however $\MM$ is an additive category, 
then one can define the Lie operad $\mathsf{Lie}\in\Op_{\{\ast\}}(\MM)$ 
with values in $\MM$ along the same lines as above.
\end{ex}

\begin{ex}\label{ex:Pois}
As in Example \ref{ex:Lie}, let us assume that $\MM$ is additive.
The {\em Poisson operad} $\mathsf{Pois}\in\Op_{\{\ast\}}(\MM)$ is the single-colored operad
presented by the following generators and relations: The generators are
\begin{subequations}\label{eqn:Poisrelationpics}
\begin{flalign}
\mu = \parbox{1cm}{\begin{tikzpicture}
\draw (0,0.5) -- (0,0);
\draw (0,0) -- (-0.4,-0.4);
\draw (0,0) -- (0.4,-0.4);
\end{tikzpicture}}
\quad,\qquad \eta = \parbox{0.5cm}{\begin{tikzpicture}
\draw (0,0.5) -- (0,0);
\draw[fill=white] (0,0) circle (0.5mm);
\end{tikzpicture}}\quad,\qquad
\{\cdot,\cdot\} = \parbox{1cm}{\begin{tikzpicture}
\draw (0,0.5) -- (0,0);
\draw (0,0) -- (-0.4,-0.4);
\draw (0,0) -- (0.4,-0.4);
\draw[fill=black] (0,0) circle (0.6mm);
\end{tikzpicture}}\quad,
\end{flalign}
where $\{\cdot,\cdot\}$ denotes the Poisson bracket.
The generators $\mu$ and $\eta$ satisfy the relations of the associative
operad \eqref{eqn:assrelationpics} and the generator $\{\cdot,\cdot\}$
the relations of the Lie operad \eqref{eqn:Lierelationpics}.
We further demand the relations
\begin{flalign}
\parbox{1cm}{\begin{tikzpicture}[baseline=21]
\draw (0,0.5) -- (0,0);
\draw (0,0) -- (-0.4,-0.4) node [below=-1pt] {\tiny{1}};
\draw (0,0) -- (0.4,-0.4) node [below=-1pt] {\tiny{2}};
\end{tikzpicture}} = \parbox{1cm}{\begin{tikzpicture}[baseline=21]
\draw (0,0.5) -- (0,0);
\draw (0,0) -- (-0.4,-0.4) node [below=-1pt] {\tiny{2}};
\draw (0,0) -- (0.4,-0.4) node [below=-1pt] {\tiny{1}};
\end{tikzpicture}}\quad,\qquad
\parbox{1cm}{\begin{tikzpicture}[baseline=21]
\draw (0,0.5) -- (0,0);
\draw (0,0) -- (-0.4,-0.4) node [below=-1pt] {\tiny{1}};
\draw (0,0) -- (0.2,-0.2);
\draw (0.2,-0.2) -- (0,-0.4) node [below=-1pt] {\tiny{2}};
\draw (0.2,-0.2) -- (0.4,-0.4) node [below=-1pt] {\tiny{3}};
\draw[fill=black] (0,0) circle (0.6mm);
\end{tikzpicture}}  = 
\parbox{1cm}{\begin{tikzpicture}[baseline=21]
\draw (0,0.5) -- (0,0);
\draw (0,0) -- (-0.2,-0.2);
\draw (0,0) -- (0.4,-0.4) node [below=-1pt] {\tiny{3}} ;
\draw (-0.2,-0.2) -- (-0.4,-0.4) node [below=-1pt] {\tiny{1}};
\draw (-0.2,-0.2) -- (0,-0.4) node [below=-1pt] {\tiny{2}};
\draw[fill=black] (-0.2,-0.2) circle (0.6mm);
\end{tikzpicture}}+
\parbox{1cm}{\begin{tikzpicture}[baseline=21]
\draw (0,0.5) -- (0,0);
\draw (0,0) -- (-0.4,-0.4) node [below=-1pt] {\tiny{2}};
\draw (0,0) -- (0.2,-0.2);
\draw (0.2,-0.2) -- (0,-0.4) node [below=-1pt] {\tiny{1}};
\draw (0.2,-0.2) -- (0.4,-0.4) node [below=-1pt] {\tiny{3}};
\draw[fill=black] (0.2,-0.2) circle (0.6mm);
\end{tikzpicture}}\quad,
\end{flalign}
\end{subequations}
which express that $\mu$ is commutative and that $\{\cdot,\cdot\}$ is a derivation in the right entry (and
hence by antisymmetry also a derivation in the left entry). Computing the  operadic composition
of the derivation relation with $\oone\otimes \eta\otimes \eta$ implies that
\begin{flalign}
\parbox{1cm}{\begin{tikzpicture}
\draw (0,0.5) -- (0,0);
\draw (0,0) -- (-0.4,-0.4);
\draw (0,0) -- (0.4,-0.4);
\draw[fill=white] (0.4,-0.4) circle (0.5mm);
\draw[fill=black] (0,0) circle (0.6mm);
\end{tikzpicture}}
=~0 \quad,
\end{flalign}
i.e.\ the Poisson bracket of the unit element is zero.
\end{ex}

\begin{ex}\label{ex:OpuLie}
This example will play a role in the formalization of linear field theories,
see Example \ref{ex:LinearFT}. Let $\MM$ be additive.
The {\em unital Lie operad} $\mathsf{uLie} \in \Op_{\{\ast\}}(\MM)$
is the single-colored operad obtained by adding to the Lie operad from Example \ref{ex:Lie}
a new generator of arity $0$, i.e.\ we have two generators
\begin{flalign}
[\cdot,\cdot] = \parbox{1cm}{\begin{tikzpicture}
\draw (0,0.5) -- (0,0);
\draw (0,0) -- (-0.4,-0.4);
\draw (0,0) -- (0.4,-0.4);
\draw[fill=black] (0,0) circle (0.6mm);
\end{tikzpicture}}\quad,\qquad
\eta = \parbox{0.5cm}{\begin{tikzpicture}
\draw (0,0.5) -- (0,0);
\draw[fill=white] (0,0) circle (0.5mm);
\end{tikzpicture}}\quad.
\end{flalign} 
In addition to the antisymmetry and Jacobi identity 
relations \eqref{eqn:Lierelationpics} for $[\cdot,\cdot]$,
we demand the compatibility relation 
\begin{flalign}
\parbox{1cm}{\begin{tikzpicture}
\draw (0,0.5) -- (0,0);
\draw (0,0) -- (-0.4,-0.4);
\draw (0,0) -- (0.4,-0.4);
\draw[fill=white] (0.4,-0.4) circle (0.5mm);
\draw[fill=black] (0,0) circle (0.6mm);
\end{tikzpicture}}
=~0 
\end{flalign}
between the Lie bracket and the unit. 
\end{ex}

We shall often require a generalization of the
concept of colored operad morphisms from Definition \ref{def:operad} to morphisms
that do not necessarily preserve the underlying sets of colors. As a preparation for the relevant definition, 
note that for every $\DDD$-colored operad 
$\P \in \Op_{\DDD}(\MM)$ and every map of sets
$f : \CCC\to\DDD$, one may define the {\em pullback $\CCC$-colored operad}
$f^\ast(\P)\in \Op_\CCC(\MM)$. Concretely, it is defined by setting 
$f^\ast(\P)\big(\substack{t \\\und{c}}\big) := \P \big(\substack{f(t)\\ f(\und{c}) }\big)$,
for all $n\geq 0$ and $(\und{c},t)\in\CCC^{n+1}$, and restricting the permutation action, 
operadic composition and operadic unit in the evident way.
\begin{defi}\label{def:varycolorOp}
The category $\Op(\MM)$ of {\em operads with varying colors} with values in $\MM$
has as objects all pairs $(\CCC,\O)$ consisting of a non-empty set
$\CCC$ and a $\CCC$-colored operad $\O\in\Op_\CCC(\MM)$. A morphism
is a pair $(f,\phi) : (\CCC,\O)\to (\DDD,\P)$ consisting
of a map of sets $f : \CCC\to \DDD$ and an $\Op_\CCC(\MM)$-morphism
$\phi : \O\to f^\ast(\P)$ to the pullback $\CCC$-colored operad.
\end{defi}

\subsection{Algebras over colored operads}
We have seen above that a colored operad
$\O$ describes abstract $n$-to-$1$ operations, for all $n\geq 0$, 
together with a composition law $\gamma$, specified identities $\oone$ 
and a permutation action $\O(\sigma)$ that allows us to permute the inputs
of operations. Forming concrete realizations/representations of these abstract
operations leads to the concept of algebras over colored operads.
\begin{defi}\label{def:Oalgebra}
An {\em algebra $A$ over a $\CCC$-colored operad $\O\in\Op_{\CCC}(\MM)$}, 
or shorter an {\em $\O$-algebra}, is  given by the following data:
\begin{itemize}
\item for each $c\in \CCC$, an object $A_c\in\MM$;
\item for each $n\geq 0$ and $(\und{c},t)\in\CCC^{n+1}$, an
$\MM$-morphism $\alpha : \O\big(\substack{ t \\ \und{c}}\big)\otimes A_{\und{c}} \to A_t$ 
(called {\em $\O$-action}), where $A_{\und{c}} := \bigotimes_{i=1}^n A_{c_i}$ with the convention
that $A_\emptyset =I$ for $n=0$.
\end{itemize}
This data is required to satisfy the standard associativity, unity and equivariance axioms, 
see e.g.\  \cite[Definition 13.2.3]{Yau}.
A morphism $\kappa : A\to B$ between two $\O$-algebras $A$ and $B$ is
a family of $\MM$-morphisms $\kappa : A_c\to B_c$, for all $c\in\CCC$, that
is compatible with the $\O$-actions, i.e.\ $\kappa\,\alpha^A  = \alpha^B\, (\id\otimes \Motimes_{i=1}^n \kappa) $.
We denote the category of $\O$-algebras by $\Alg_{\O}$.
\end{defi}

\begin{ex}\label{ex:Diagalgebra}
Consider the diagram operad $\Diag_{\CC}\in \Op_{\CC_0}(\Set)$
from Example \ref{ex:Diag}. A $\Diag_{\CC}$-algebra 
is a family of sets $A_c\in\Set$, for all objects $c\in \CC_0$ in the category $\CC$,
together with maps $\alpha : \Diag_{\CC}\big(\substack{t\\ c}\big) \times A_c\to A_t$,
for all $c,t\in \CC_0$. (Here we already used that $\Diag_{\CC}$ only contains $1$-ary operations.)
Because $\Diag_{\CC}\big(\substack{t\\ c}\big) = \CC(c,t)$ is the $\Hom$-set, 
the latter data is equivalent to specifying for each $\CC$-morphism $f:c\to t$
a map of sets $A(f):= \alpha(f,-) : A_c\to A_t$. The axioms for $\O$-algebras
imply that $A(g\,f) = A(g) \, A(f)$, for all composable
$\CC$-morphism, and $A(\id) = \id$ for the identities. Hence,
a $\Diag_{\CC}$-algebra is precisely a functor $\CC\to \Set$, i.e.\ a diagram
of shape $\CC$. One observes that morphisms
between $\Diag_\CC$-algebras are precisely 
natural transformations between the
corresponding functors.
\end{ex}

\begin{ex}\label{ex:Assalgebra}
Consider for the moment $\MM=\Set$ and the associative 
operad $\mathsf{As}\in\Op_{\{\ast\}}(\Set)$ from  Example \ref{ex:Ass}.
An $\mathsf{As}$-algebra is a single set $A = A_\ast\in\Set$
together with an $\mathsf{As}$-action. The latter is equivalent to providing a family of maps
$\alpha : \mathsf{As}(n)  \to \mathrm{Map}(A^{\times n},A)$,
for all $n\geq 0$, which define an $\Op_{\{\ast\}}(\Set)$-morphism to the endomorphism operad $\mathrm{End}(A)$,
see e.g.\ \cite[Definition 13.8.1]{Yau}. 
Because $\mathsf{As}$ is presented by generators and relations (see Example \ref{ex:Ass}),
this is equivalent to defining $\alpha$ on the generators such that the relations hold true.
This yields two maps $\mu_A := \alpha(\mu) : A\times A\to A$ and $\eta_A:=\alpha(\eta) : \{\ast\} \to A$,
which because of the relations have to satisfy the axioms of an associative and unital algebra in $\Set$.
One finds that morphisms
of $\mathrm{As}$-algebras are precisely morphisms of associative and unital algebras.
\sk

For a general bicomplete closed symmetric monoidal category $\MM$,
one obtains that the category $\Alg_{\mathsf{As}}$ of algebras over 
$\mathsf{As}\in\Op_{\{\ast\}}(\MM)$ is the category of associative and unital algebras 
in $\MM$. In particular, for $\MM=\Vec_\bbK$ this is the category of associative 
and unital $\bbK$-algebras.
\end{ex}

\begin{ex}\label{ex:LieandPoisalgebra}
A similar argument as in Example \ref{ex:Assalgebra} shows that
the category $\Alg_{\mathsf{Lie}}$ of algebras over the Lie operad $\mathsf{Lie}\in \Op_{\{\ast\}}(\MM)$
(see Example \ref{ex:Lie}) is the category of Lie algebras in $\MM$ and that the
category $\Alg_{\mathsf{Pois}}$ of algebras over the Poisson 
operad $\mathsf{Pois}\in\Op_{\{\ast\}}(\MM)$ (see Example \ref{ex:Pois}) is
the category of Poisson algebras in $\MM$. 
\end{ex}

Given an $\Op(\MM)$-morphism $(f,\phi): (\CCC,\O)\to (\DDD,\P)$ in the sense of Definition \ref{def:varycolorOp},
one may define a pullback functor $(f,\phi)^\ast : \Alg_{\P}\to \Alg_{\O}$ between the corresponding
categories of algebras. The pullback of $A\in\Alg_{\P}$ is the $\O$-algebra
defined by $((f,\phi)^\ast A)_c := A_{f(c)}\in\MM$, for all $c\in \CCC$, together with
the $\O$-action
\begin{flalign}
\xymatrix@C=3em{
\O\big(\substack{t\\\und{c}}\big)\otimes ((f,\phi)^\ast A)_{\und{c}} \ar[r]^-{\phi\otimes \id}&
\P\big(\substack{f(t) \\ f(\und{c})}\big)\otimes A_{f(\und{c})} \ar[r]^-{\alpha}&
A_{f(t)} = ((f,\phi)^\ast A)_t
}\quad.
\end{flalign}
\begin{theo}\label{theo:adjunction}
For any $\Op(\MM)$-morphism $(f,\phi): (\CCC,\O)\to (\DDD,\P)$, the pullback functor 
$(f,\phi)^\ast : \Alg_{\P}\to \Alg_{\O}$ has a left adjoint, 
which is called {\em operadic left Kan extension}.
We denote the corresponding adjunction by
\begin{flalign}
\xymatrix{
(f,\phi)_! \,:\, \Alg_{\O} ~\ar@<0.5ex>[r]&\ar@<0.5ex>[l]  ~\Alg_{\P} \,:\, (f,\phi)^\ast
}\quad.
\end{flalign}
\end{theo}

\begin{ex}
Every functor $F: \CC\to \DD$ defines
an evident $\Op(\Set)$-morphism $(F_0, F) : (\CC_0,\Diag_\CC)\to (\DD_0,\Diag_\DD)$ 
between the corresponding diagram operads (see Example \ref{ex:Diag}).
Recalling from Example \ref{ex:Diagalgebra} that
$\Alg_{\Diag_\CC}\cong \Set^\CC$ is the category of functors from 
$\CC$ to $\Set$ (and similarly that $\Alg_{\Diag_\DD}\cong \Set^\DD$), one shows
that the pullback functor $(F_0,F)^\ast $ is the usual pullback functor 
$F^\ast :=(-)\circ F : \Set^\DD\to \Set^\CC$ for functor categories.
Its left adjoint $(F_0,F)_!$ is therefore the ordinary categorical 
left Kan extension $\Lan_F : \Set^\CC\to \Set^\DD$. 
\end{ex}


\section{\label{sec:FTs}Field theory operads}
\subsection{Orthogonal categories and field theories}
Let us briefly recall the basic idea of algebraic quantum field theory, 
see e.g.\ \cite{HaagKastler,Brunetti,AQFTbook,FewsterVerch}
for more details. Broadly speaking, a field theory in this setting is a functor
from a suitable category of spacetimes to a category of algebraic structures of interest
that satisfies a list of physically motivated axioms. The prime example is given by functors
$\AAA: \Loc\to \Alg_{\mathsf{As}}$ from the category $\Loc$ of globally hyperbolic Lorentzian
manifolds to the category of associative and unital algebras 
that satisfy the Einstein causality axiom. The latter is a property
of the functor $\AAA: \Loc\to \Alg_{\mathsf{As}}$ which demands
that for every pair $(f_1 : M_1 \to M, f_2: M_2 \to M)$ of $\Loc$-morphisms
whose images are causally disjoint in $M$ the diagram
\begin{flalign}\label{eqn:EinsteinCausality}
\xymatrix@C=5em{
\ar[d]_-{\AAA(f_1)\otimes \AAA(f_2)} \AAA(M_1)\otimes \AAA(M_2) \ar[r]^-{\AAA(f_1)\otimes \AAA(f_2)} & \AAA(M)\otimes \AAA(M)\ar[d]^-{\mu_M^\op}\\
\AAA(M)\otimes\AAA(M) \ar[r]_-{\mu_M^{}} &\AAA(M) 
}
\end{flalign}
commutes, where $\mu_M^{(\op)}$ denotes the (opposite) multiplication on $\AAA(M)$.
Another important physical example is given by functors
$\AAA: \Loc\to \Alg_{\mathsf{As}}$ that satisfy the time-slice axiom in addition to the Einstein causality axiom.
Such theories will be formalized later in Section \ref{subsec:orthlocalication}
via localization techniques.
\sk

For the purpose of this paper, we consider the following generalization 
of quantum field theories satisfying the Einstein causality axiom.
(Examples which justify this generalization are presented at the end of this subsection.)
Let $\CC$ be a small category which we interpret as a category of spacetimes. Instead of 
associative and unital algebras, let us take any single-colored operad $\P\in\Op_{\{\ast\}}(\MM)$
and consider the functor category ${\Alg_{\P}}^\CC$. 
An object in this category 
is a functor $\AAA : \CC\to\Alg_{\P}$, i.e.\ an assignment of $\P$-algebras to 
spacetimes, and the morphisms are natural transformations
between such functors. To encode physical axioms which generalize the Einstein causality axiom above, 
we recall the concept of orthogonal categories from \cite{BeniniSchenkelWoike}. 
\begin{defi}\label{def:orthogonalcat}
An {\em orthogonal category} is a pair $\ovr{\CC} := (\CC, \perp)$ consisting
of a small category $\CC$ and a subset $\perp 
\subseteq \mathrm{Mor}\,\CC\, {}_{\mathrm{t}}{\times}_{\mathrm{t}}\,\mathrm{Mor}\,\CC$ 
of the set of pairs of morphisms with a common target, which satisfies the
following properties:
\begin{itemize}
\item[(i)] {\em Symmetry:} If $(f_1,f_2)\in \perp$,
then $(f_2,f_1)\in\perp $.

\item[(ii)] {\em Stability under compositions:} If $(f_1,f_2)\in \perp$, then
$(g\, f_1\, h_1 ,g\, f_2\, h_2) \in\perp$ for all composable $\CC$-morphisms $g$, $h_1$ and $h_2$.
\end{itemize}
We shall also write $f_1\perp f_2$ for elements $(f_1,f_2)\in \perp$.
An {\em orthogonal functor} $F : \ovr{\CC}\to\ovr{\DD}$ is a functor $F : \CC\to \DD$ 
that preserves orthogonality, i.e.\ if $f_1\perp_\CC f_2$ then $F(f_1)\perp_\DD F(f_2)$.
We denote by $\OCat$ the category of orthogonal categories and orthogonal functors.
\end{defi}

\begin{ex}\label{ex:Loc}
Let $\Loc$ be any small category that is equivalent to the usual category 
of oriented, time-oriented and globally hyperbolic Lorentzian spacetimes 
of a fixed dimension $\geq 2$, see \cite{Brunetti,FewsterVerch}. 
We define $\perp_\Loc$ as the subset of all pairs $(f_1 : M_1\to M, f_2: M_2\to M)$
of causally disjoint $\Loc$-morphisms, i.e.\ pairs of morphisms 
such that the images $f_1(M_1)$ and $f_2(M_2)$
are causally disjoint subsets in $M$.
The pair $\ovr{\Loc} := (\Loc,\perp_{\Loc})$ defines an orthogonal category.
\end{ex}

Let us now consider a parallel pair of $\Seq_{\{\ast\}}(\MM)$-morphisms
$r_1,r_2 : I[2] \rightrightarrows U(\P)$, where $I[2]\in \Seq_{\{\ast\}}(\MM)$ is defined by
\begin{flalign}\label{eqn:I2}
I[2](n)\, :=\, \begin{cases}
I & ~,~~\text{for } n=2~,\\
\emptyset &~,~~\text{else },
\end{cases}
\end{flalign}
for all $n\geq 0$. This means that each $r_i$ picks out an operation of arity $2$ in $\P$.
For simplifying notation, we shall write 
\begin{flalign}
\P^{(r_1,r_2)} \,:=\, \big(\P, r_1,r_2 : I[2] \rightrightarrows U(\P)\big)
\end{flalign}
and we call $\P^{(r_1,r_2)}$ an {\em (arity $2$) bipointed single-colored operad}.
\begin{defi}\label{def:fieldtheory}
A {\em field theory of type $\P^{(r_1,r_2)}$ on $\ovr{\CC}$}  is a functor
$\AAA : \CC \to\Alg_{\P}$ that satisfies the following property:
For all $(f_1 : c_1\to c)\perp (f_2:c_2\to c)$, the diagram
\begin{flalign}\label{eqn:FTproperty}
\xymatrix@C=8em{
\ar[d]_-{r_2 \otimes\AAA(f_1)\otimes \AAA(f_2) } I\otimes \AAA(c_1) \otimes \AAA(c_2)  \ar[r]^-{r_1 \otimes\AAA(f_1)\otimes \AAA(f_2) } & \P(2)\otimes \AAA(c)^{\otimes 2}  \ar[d]^-{\alpha^\P_c}\\
 \P(2)\otimes \AAA(c)^{\otimes 2} \ar[r]_-{\alpha^\P_c} & \AAA(c)
}
\end{flalign}
in $\MM$ commutes, where $\alpha^\P_c $ denotes the $\P$-action on $\AAA(c) \in\Alg_\P$
(see Definition \ref{def:Oalgebra}). The {\em category of field theories 
of type $\P^{(r_1,r_2)}$ on $\ovr{\CC}$} is defined as the full subcategory
\begin{flalign}
\FT\big(\ovr{\CC},\P^{(r_1,r_2)}\big) \,\subseteq\, {\Alg_\P}^\CC\quad,
\end{flalign}
whose objects are all functors $\AAA : \CC\to\Alg_{\P}$ satisfying \eqref{eqn:FTproperty}.
\end{defi}

\begin{rem}
Our concept of field theories in Definition \ref{def:fieldtheory} is based
on the idea that there exist two distinguished arity $2$ operations in $\P$, 
which act in the same way when pre-composed with an orthogonal pair $f_1 \perp f_2$
of $\CC$-morphisms. There exists an obvious generalization of this scenario
to $n$-ary operations in $\P$ and orthogonal $n$-tuples of $\CC$-morphisms.
We however decided not to introduce this more general framework for field theories,
because all examples of interest to us are field theories in the sense of 
Definition \ref{def:fieldtheory}. 
\end{rem}

\begin{ex}[Quantum field theories]\label{ex:LCQFT}
Consider the associative operad $\mathsf{As}\in \Op_{\{\ast\}}(\MM)$ from Example
\ref{ex:Ass} and the two $\Seq_{\{\ast\}}(\MM)$-morphisms 
$\mu,\mu^\op : I[2] \rightrightarrows U(\mathsf{As})$ 
which select the multiplication and opposite multiplication operations.
A field theory of type $\mathsf{As}^{(\mu,\mu^\op)}$ on $\ovr{\CC}$
is a functor $\AAA : \CC\to \Alg_{\mathsf{As}}$ to the category of associative and unital
algebras which satisfies the analog of \eqref{eqn:EinsteinCausality}. For 
$\ovr{\CC} = \ovr{\Loc}$ (see Example \ref{ex:Loc}), 
this is a locally covariant quantum field theory \cite{Brunetti,FewsterVerch}
that satisfies the Einstein causality axiom but not necessarily the time-slice axiom.
The time-slice axiom will be discussed in Section \ref{subsec:orthlocalication}.
\end{ex}

\begin{rem}\label{rem:LCQFTviacommutator}
If $\MM$ is additive, there exists an alternative but equivalent
formalization of the type of field theories from Example \ref{ex:LCQFT}. Consider
the associative operad $\mathsf{As}\in \Op_{\{\ast\}}(\MM)$ and the two 
$\Seq_{\{\ast\}}(\MM)$-morphisms $[\cdot,\cdot],0 : I[2] \rightrightarrows U(\mathsf{As})$ 
which select the commutator $[\cdot,\cdot]= \mu-\mu^\op$ and the zero-operation (of arity $2$).
A field theory of type $\mathsf{As}^{([\cdot,\cdot],0)}$ on $\ovr{\CC}$ 
is a functor $\AAA : \CC\to\Alg_{\mathsf{As}}$ to the category of associative and unital 
algebras which satisfies the property that
\begin{flalign}\label{eqn:CommutatorEinsteinCausality}
\big[ \AAA(f_1)(-), \AAA(f_2)(-)\big]_c \,:\, \AAA(c_1)\otimes\AAA(c_2) \longrightarrow \AAA(c)
\end{flalign}
is the zero-map, for all $(f_1:c_1\to c ) \perp (f_2:c_2\to c)$. (Here $[\cdot,\cdot]_c = \mu_c - \mu_c^\op$ 
denotes the commutator on $\AAA(c)$.) This is equivalent to our description in Example \ref{ex:LCQFT},
i.e.\ 
\begin{flalign}\label{eqn:LCQFTviacommutator}
\FT\big(\ovr{\CC},\mathsf{As}^{([\cdot,\cdot],0)}\big) ~\cong~ \FT\big(\ovr{\CC}, \mathsf{As}^{(\mu,\mu^\op)}\big)
\quad.
\end{flalign}
This observation will be useful in Section \ref{sec:quantization} 
when we study the linear quantization adjunction.
\end{rem}

\begin{ex}[Classical field theories]\label{ex:ClassicalFT}
Let $\MM$ be additive. Consider the Poisson operad $\mathsf{Pois}\in \Op_{\{\ast\}}(\MM)$
from Example \ref{ex:Pois} and the two $\Seq_{\{\ast\}}(\MM)$-morphisms 
$\{\cdot,\cdot\},0 : I[2] \rightrightarrows U(\mathsf{Pois})$ 
which select the Poisson bracket and the zero-operation. A field theory
of type $\mathsf{Pois}^{(\{\cdot,\cdot\},0)}$ on $\ovr{\CC}$ 
is a functor $\AAA : \CC\to\Alg_{\mathsf{Pois}}$ to the category of Poisson 
algebras which satisfies the property that
\begin{flalign}\label{eqn:PoissonEinsteinCausality}
\big\{ \AAA(f_1)(-), \AAA(f_2)(-)\big\}_c \,:\, \AAA(c_1)\otimes\AAA(c_2) \longrightarrow \AAA(c)
\end{flalign}
is the zero-map, for all $(f_1:c_1\to c ) \perp (f_2:c_2\to c)$. (Here $\{\cdot,\cdot\}_c$ denotes the Poisson
bracket on $\AAA(c)$.) 
For $\ovr{\CC} = \ovr{\Loc}$, this is a classical analog of locally covariant quantum field theory, 
where one assigns to each spacetime a Poisson algebra of classical observables, see e.g.\
\cite{BrunettiFredenhagenRibeiro,FredenhagenRejzner,Collini,BeniniSchenkelPoisson}.
The property \eqref{eqn:PoissonEinsteinCausality} 
demands that the Poisson bracket between causally disjoint classical observables is zero, 
which captures the classical analog of the Einstein causality axiom. 
\end{ex}

\begin{ex}[Linear field theories]\label{ex:LinearFT}
In the usual construction of linear quantum field theories,
see e.g.\ \cite{Baer,BGproc,BeniniDappiaggiHack} for reviews, one first defines
a functor $\LLL : \Loc \to \mathbf{PSymp}$ to the category of presymplectic vector spaces,
which is then quantized by forming CCR-algebras (CCR stands for canonical commutation relations).
Recall that a presymplectic vector space $(V,\omega)$ is a pair consisting of 
a vector space $V$ and an antisymmetric linear map $\omega : V\otimes V\to\bbK$.
Notice that this is {\em not} an operation of arity $2$ in the sense of operads 
because the target is the ground field and not $V$. Hence,
$\mathbf{PSymp}$ is {\em not} the category of algebras over an operad and, as a consequence,
functors $\LLL : \Loc \to \mathbf{PSymp}$ do {\em not} define field theories in the sense 
of Definition \ref{def:fieldtheory}.
\sk

However, there exists a canonical upgrade of every functor $\LLL : \Loc \to \mathbf{PSymp}$ 
to a field theory in the sense of Definition \ref{def:fieldtheory}. 
Given any presymplectic vector space
$(V,\omega)$, one can define its {\em Heisenberg Lie algebra} $H(V,\omega)$.
The underlying vector space of $H(V,\omega)$ is given by $V \oplus \bbK$
and the Lie bracket $[-,-] : (V \oplus \bbK)\otimes (V \oplus \bbK)\to V \oplus \bbK$ is
\begin{flalign}
[v\oplus k, v^\prime\oplus k^\prime ] \,:=\, 0 \oplus \omega(v,v^\prime)\quad,
\end{flalign}
for all $v\oplus k, v^\prime\oplus k^\prime \in V\oplus \bbK$. There exists a canonical unit map 
$\eta : \bbK \to V\oplus \bbK\,,~k\mapsto 0\oplus k$, which is compatible
with the Lie bracket, i.e.\ $[v\oplus k, \eta(k^\prime)] =0$, 
for all $v\oplus k \in V\oplus \bbK $ and  $k^\prime\in\bbK$.
Hence, Heisenberg Lie algebras are algebras over the unital Lie operad
$\mathsf{uLie}\in \Op_{\{\ast\}}(\MM)$ given in  Example \ref{ex:OpuLie}.
Because forming Heisenberg Lie algebras is functorial,
we can define for every $\LLL : \Loc \to \mathbf{PSymp}$  the composite functor 
$H\,\LLL : \Loc \to \Alg_{\mathsf{uLie}}$.
\sk

Consider now the two $\Seq_{\{\ast\}}(\MM)$-morphisms
$[\cdot,\cdot],0 : I[2] \to U(\mathsf{uLie})$ which select the Lie bracket
and the zero-operation. A field theory of type
${\mathsf{uLie}}^{([\cdot,\cdot],0)}$ on $\ovr{\CC}$
is a functor $\AAA : \CC\to\Alg_{\mathsf{uLie}}$ to the category of unital Lie 
algebras which satisfies the property that
\begin{flalign}\label{eqn:LieEinsteinCausality}
\big [ \AAA(f_1)(-), \AAA(f_2)(-)\big ]_c \,:\, \AAA(c_1)\otimes\AAA(c_2) \longrightarrow \AAA(c)
\end{flalign}
is the zero-map, for all $(f_1:c_1\to c ) \perp (f_2:c_2\to c)$. (Here $[\cdot,\cdot]_c$ denotes the Lie
bracket on $\AAA(c)$.) This property is a suitable analog of the Einstein
causality axiom for linear field theories. In particular, if $\ovr{\CC} = \ovr{\Loc}$,
$\MM = \Vec_\bbK$ and $\AAA =H\,\LLL : \Loc \to \Alg_{\mathsf{uLie}}$ is given by applying
the Heisenberg Lie algebra construction to a functor $\LLL : \Loc \to\mathbf{PSymp}$ 
with values in presymplectic vector spaces, then \eqref{eqn:LieEinsteinCausality}
is equivalent to the property that the presymplectic structure of causally disjoint 
linear observables is zero. This is precisely the Einstein causality axiom for linear field 
theories, see e.g.\ \cite{Baer,BGproc,BeniniDappiaggiHack}. 
\end{ex}

\subsection{Operadic description}
In this section we show that the category
of field theories from Definition \ref{def:fieldtheory}
is the category of algebras over a suitable colored operad.
This generalizes previous results in \cite{BeniniSchenkelWoike}
and it is the key insight that allows us to study a large family of
universal constructions for field theories in Section \ref{sec:adjunctions}.
As a preparation for the relevant definition,
we define an auxiliary colored operad that describes 
functors from a small category $\CC$ to the category of 
$\P$-algebras. 
\begin{defi}\label{def:PCoperad}
Let $\CC$ be a small category with set of objects $\CC_0$  
and let $\P\in \Op_{\{\ast\}}(\MM)$ be a single-colored operad.
The $\CC_0$-colored operad $\P_\CC\in \Op_{\CC_0}(\MM)$ is defined
by the following data:
\begin{itemize}
\item for $n\geq 0$ and $(\und{c},t)\in \CC_0^{n+1}$, the object of operations is
\begin{flalign}\label{eqn:PCoperations}
\P_\CC\big(\substack{t \\ \und{c}}\big)\,:=\, \CC(\und{c},t) \otimes \P(n) \,\in\,\MM\quad,
\end{flalign}
where $\otimes$ is the $\Set$-tensoring \eqref{eqn:Settensoring} and
$\CC(\und{\cc},t):=\prod_{i=1}^n \CC(c_i,t)$ is the product of $\Hom$-sets;

\item for $n\geq 0$, $(\und{c},t)\in\CC_0^{n+1}$ and $\sigma \in\Sigma_n$, 
the permutation action $\P_\CC(\sigma)$ is defined by
\begin{flalign}
\xymatrix@C=4em{
\P_\CC\big(\substack{t \\ \und{c}}\big) \ar[r]^-{\P_\CC(\sigma)}&\P_\CC\big(\substack{t \\ \und{c}\sigma}\big) \\
\ar[u]^-{\iota_{\und{f}}}\P(n)\ar[r]_-{\P(\sigma) } & \P(n)\ar[u]_-{\iota_{\und{f}\sigma}}
}\quad
\end{flalign}
for all $\und{f} := (f_1,\dots,f_n)\in\CC(\und{c},t)$, where 
$\iota_{\und{f}} : \P(n)\to \P_\CC\big(\substack{t \\ \und{c}}\big) = \CC(\und{c},t) \otimes \P(n)$ 
are the inclusion morphisms into the coproduct (see \eqref{eqn:Settensoring})
and $\und{f}\sigma := (f_{\sigma(1)},\dots,f_{\sigma(n)})$;

\item for $n>0$, $k_1,\dots,k_n \geq 0$, $(\und{a},t)\in\CC_0^{n+1}$ 
and $(\und{b}_i,a_i)\in\CC_0^{k_i+1}$, for $i=1,\dots,n$, the operadic
composition $\gamma^{\P_\CC}$ is defined by
\begin{flalign}
\xymatrix@C=4em{
\P_\CC\big(\substack{t \\ \und{a}} \big) \otimes \bigotimes\limits_{i=1}^n \P_\CC\big(\substack{a_i \\ \und{b}_i} \big)
\ar[r]^-{\gamma^{\P_\CC}} &
\P_\CC\big(\substack{t \\ \und{b}}\big)\\
\ar[u]^-{\iota_{\und{f}}\,\otimes\, \Motimes_{i=1}^n \iota_{\und{g}_i}} \P(n)\otimes \bigotimes\limits_{i=1}^n \P(k_i) \ar[r]_-{\gamma^\P} & \P(k_1+\cdots + k_n)\ar[u]_-{\iota_{\und{f}(\und{g}_1,\dots,\und{g}_n)}}
}\quad
\end{flalign}
for all $\und{f} = (f_1,\dots,f_n)\in\CC(\und{a},t)$ and $\und{g}_i =(g_{i1},\dots,g_{i k_i})\in\CC(\und{b}_i,a_i)$, 
for $i=1,\dots,n$, where $\und{f}(\und{g}_1,\dots,\und{g}_n) := (f_1 \,g_{11},\dots ,f_{n} \, 
g_{n k_n} )\in \CC(\und{b},t)$ is defined by composition in the category $\CC$;

\item  for $c\in\CC_0$, the operadic unit $\oone^{\P_\CC}$ is
\begin{flalign}
\xymatrix@C=4em{
\ar[rd]_-{\oone^\P} I \ar[r]^-{\oone^{\P_\CC}} & \P_\CC\big(\substack{c \\ c}\big)\\
&\P(1)\ar[u]_-{\iota_{\id_c}}
}\quad
\end{flalign}
where $\id_c : c\to c$ is the identity morphism of $c$ in the category $\CC$.
\end{itemize}
A straightforward check shows that this data defines a colored operad
(see Definition \ref{def:operad}).
\end{defi}

\begin{lem}\label{lem:PCalgebras}
There exists a canonical isomorphism
\begin{flalign}
\Alg_{\P_\CC}~\cong~{\Alg_{\P}}^\CC
\end{flalign}
between the category of algebras over the colored operad
$\P_\CC \in \Op_{\CC_0}(\MM)$ from Definition \ref{def:PCoperad}
and the category of functors from $\CC$ to $\Alg_{\P}$.
\end{lem} 
\begin{proof}
A $\P_\CC$-algebra is a family of
objects $A_c\in\MM$, for all $c\in\CC_0$, together with a $\P_\CC$-action
$\alpha : \P_\CC\big(\substack{t \\ \und{c}}\big)\otimes A_{\und{c}}\to A_t$.
Because \eqref{eqn:PCoperations} is a coproduct, this is
equivalent to a family of $\MM$-morphisms
$\alpha_{\und{f}} : \P(n) \otimes A_{\und{c}} \to A_t$,
for all $n\geq 0$, $(\und{c},t)\in\CC_0^{n+1}$ and $\und{f}\in\CC(\und{c},t)$,
which satisfies the following compatibility conditions 
resulting from the axioms for algebras over colored operads
\begin{subequations}\label{eqn:PCaction}
\begin{flalign}\label{eqn:PCaction1}
\xymatrix@C=4em{
\ar[rr]^-{\gamma^\P\otimes\id} 
\Big(\P(n)\otimes\bigotimes\limits_{i=1}^n \P(k_i)\Big)\otimes A_{\und{b}} 
\ar[d]_-{\text{permute}}^-{\cong}
& 
&  \P(k_1+\cdots+k_n) \otimes A_{\und{b}}
\ar[d]^-{\alpha_{\und{f}(\und{g}_1,\dots,\und{g}_n)}}
\\
\P(n)\otimes\bigotimes\limits_{i=1}^n \Big( \P(k_i)\otimes A_{\und{b}_i}\Big)  
\ar[r]_-{\id\otimes \Motimes_i \alpha_{\und{g}_i}}
&\P(n) \otimes A_{\und{a}}
\ar[r]_-{\alpha_{\und{f}}}
&A_t
}
\end{flalign}
\begin{flalign}\label{eqn:PCaction2}
\xymatrix@C=3.5em{
\ar[dr]_-{\cong} I\otimes A_c \ar[r]^-{\oone^\P\otimes\id}  &\P(1)\otimes A_c\ar[d]^-{\alpha_{\id_c}} & &
\ar[d]_-{\P(\sigma) \otimes \,\text{permute}}\P(n) \otimes A_{\und{c}} \ar[r]^-{\alpha_{\und{f}}}& A_t\\
& A_c & & \P(n) \otimes A_{\und{c}\sigma}  \ar[ru]_-{~~\alpha_{\und{f}\sigma}} &
}
\end{flalign}
\end{subequations}
Using that any $\und{f} = (f_1,\dots,f_n)\in\CC(\und{c},t)$ can be written as
$\und{f} = \und{\id_t}^n(f_1,\dots,f_n)$, where $\und{\id_t}^n = (\id_t,\dots,\id_t)$ is of length $n$,
the diagram \eqref{eqn:PCaction1} implies that $\alpha_{\und{f}}$ factorizes as
\begin{flalign}
\xymatrix@C=3.5em{
\ar[d]_-{\id\otimes\Motimes_i (\oone^\P\otimes\id )}\P(n)\otimes\bigotimes\limits_{i=1}^n \Big(I \otimes A_{c_i}\Big)~\cong~\P(n) \otimes A_{\und{c}}\ar[r]^-{\alpha_{\und{f}}} & A_t\\
\P(n)\otimes\bigotimes\limits_{i=1}^n \Big(\P(1)\otimes A_{c_i}\Big) \ar[r]_-{\id\otimes\Motimes_i \alpha_{f_i}}& \P(n)\otimes A_t^{\otimes n} \ar[u]_-{\alpha_{\und{\id_t}^n}}
} 
\end{flalign}
Hence, the $\P_\CC$-action $\alpha$ is uniquely specified by
the following two types of $\MM$-morphisms: (1)~$\widetilde{\alpha}_t := \alpha_{\und{\id_t}^n}:
\P(n)\otimes A_t^{\otimes n} \to A_t$, for all $t\in \CC_0$ and $n\geq 0$, and 
(2)~$A(f) := \alpha_f~(\oone^\P\otimes \id) :  A_c \,\cong\, I\otimes A_c \to A_t$,
for all $\CC$-morphisms $f:c\to t$. The remaining conditions in \eqref{eqn:PCaction}
are equivalent to $\widetilde{\alpha}_t $ defining a $\P$-action on $A_t$, for all $t\in\CC_0$, 
and $A(f) : A_c\to A_t$ defining a functor $\CC\to \Alg_\P$ to $\P$-algebras.
From this perspective, $\P_\CC$-algebra morphisms correspond precisely
to natural transformations between functors from $\CC$ to $\Alg_\P$.
\end{proof}

For the rest of this subsection we fix an orthogonal category $\ovr{\CC}=(\CC,\perp)$
and a bipointed single-colored operad $\P^{(r_1,r_2)}= (\P, r_1,r_2 : I[2] \rightrightarrows U(\P) )$.
(Recall the definition of $I[2]$ in \eqref{eqn:I2}.)
We define a $\CC_0$-colored sequence $R_\perp \in \Seq_{\CC_0}(\MM)$ 
by setting 
\begin{flalign}
R_\perp \big(\substack{t \\ \und{c}}\big) \,:=\, \begin{cases}
\big({\perp} \, {\cap}\, \CC(\und{c},t)\big)\otimes I & ~,~~\text{for }n=2~,\\
\emptyset & ~,~~\text{else }, 
\end{cases}
\end{flalign}
for all $n\geq 0$ and $(\und{c},t)\in\CC_0^{n+1}$, and a parallel pair
of $\Seq_{\CC_0}(\MM)$-morphisms 
\begin{subequations}\label{eqn:Rperprelations}
\begin{flalign}
\xymatrix@C=2.5em{
r_{1,\CC},r_{2,\CC} \,:\, R_\perp \ar@<-0.5ex>[r]\ar@<0.5ex>[r]\,&\, U(\P_\CC)
}
\end{flalign} 
by setting, for $i=1,2$,
\begin{flalign}
\xymatrix@C=4em{
R_\perp \big(\substack{t \\ (c_1,c_2) }\big) \ar[r]^-{r_{i,\CC}} & \P_\CC\big(\substack{t \\ (c_1,c_2) }\big)\\
\ar[u]^-{\iota_{(f_1,f_2)}} I \ar[r]_-{r_i} & \P(2)\ar[u]_-{\iota_{(f_1,f_2)}}
}
\end{flalign}
\end{subequations}
for all $(f_1:c_1\to t, f_2: c_2\to t)\in \perp$.
\begin{defi}\label{def:FToperad}
The {\em operad of field theories of type $\P^{(r_1,r_2)}$ on $\ovr{\CC}$}
is defined as the coequalizer
\begin{flalign}\label{eqn:FToperad}
\xymatrix@C=2.5em{ 
F(R_\perp) \ar@<-0.5ex>[r]_-{r_{2,\CC}}\ar@<0.5ex>[r]^-{r_{1,\CC}}  ~&~\P_\CC \ar@{-->}[r]  ~&~  
\P^{(r_1,r_2)}_{\ovr{\CC}} 
}
\end{flalign}
in $\Op_{\CC_0}(\MM)$.
\end{defi}

The importance of this operad is evidenced by the following theorem. 
\begin{theo}\label{theo:FTcatiso}
There exists a canonical isomorphism
\begin{flalign}
\Alg_{\P^{(r_1,r_2)}_{\ovr{\CC}}} ~\cong~ \FT\big(\ovr{\CC},\P^{(r_1,r_2)}\big)
\end{flalign}
between the category of algebras over the colored operad
$\P^{(r_1,r_2)}_{\ovr{\CC}} \in \Op_{\CC_0}(\MM)$ from Definition \ref{def:FToperad}
and the category of field theories of type $\P^{(r_1,r_2)}$ on $\ovr{\CC}$ from Definition \ref{def:fieldtheory}.
\end{theo}
\begin{proof}
Because $\P^{(r_1,r_2)}_{\ovr{\CC}}$ is defined as a coequalizer \eqref{eqn:FToperad},
its algebras are precisely those $\P_\CC$-algebras $A\in\Alg_{\P_\CC}$ that satisfy the relations encoded by 
$r_{1,\CC},r_{2,\CC} : R_\perp \rightrightarrows U(\P_\CC)$ (see \eqref{eqn:Rperprelations}). 
Using the notations from the proof of Lemma \ref{lem:PCalgebras}, this concretely
means that the diagram
\begin{flalign}
\xymatrix@C=3.5em{
\ar[d]_-{r_2\otimes\id\otimes\id}I\otimes A_{c_1}\otimes A_{c_2} \ar[r]^-{r_1\otimes \id\otimes\id} & \P(2)\otimes A_{c_1}\otimes A_{c_2} \ar[d]^-{\alpha_{(f_1,f_2)}}\\
 \P(2)\otimes A_{c_1}\otimes A_{c_2} \ar[r]_-{\alpha_{(f_1,f_2)}}& A_t
}
\end{flalign}
in $\MM$ commutes, for all $(f_1: c_1\to t, f_2: c_2\to t)\in\perp$. Using the isomorphism
of Lemma  \ref{lem:PCalgebras}, one easily translates this diagram
to the diagram \eqref{eqn:FTproperty} for the functor $\AAA : \CC\to\Alg_\P$ 
corresponding to $A\in\Alg_{\P_\CC}$, which completes the proof.
\end{proof}

\begin{ex}
Recalling Examples \ref{ex:LCQFT}, \ref{ex:ClassicalFT} and \ref{ex:LinearFT},
our construction defines colored operads for quantum field theory $\mathsf{As}^{(\mu,\mu^\op)}_{\ovr{\CC}}$
(or equivalently $\mathsf{As}^{([\cdot,\cdot],0)}_{\ovr{\CC}}$ provided that $\MM$ is additive, 
see Remark \ref{rem:LCQFTviacommutator}),
for classical field theory $\mathsf{Pois}^{(\{\cdot,\cdot\},0)}_{\ovr{\CC}}$ and for linear field theory
$\mathsf{uLie}^{([\cdot,\cdot],0)}_{\ovr{\CC}}$ formalized in terms of Heisenberg Lie algebras.
\end{ex}

\subsection{Functoriality}
Note that the field theory operad $\P^{(r_1,r_2)}_{\ovr{\CC}}\in \Op_{\CC_0}(\MM)$ 
from Definition \ref{def:FToperad} depends on the choice of two kinds of data:
(1) An orthogonal category $\ovr{\CC} = (\CC,\perp)$ and (2) a bipointed single-colored operad
$\P^{(r_1,r_2)}= (\P, r_1,r_2 : I[2] \rightrightarrows U(\P) )$.
We will see that both of these dependencies are functorial.
Recall from Definition \ref{def:orthogonalcat} that orthogonal categories
are the objects of the category $\OCat$. The second kind of data
may be arranged in terms of a category as follows.
\begin{defi}\label{def:bipointed}
The {\em category of (arity $2$) bipointed single-colored operads} $\Op_{\{\ast\}}^{2\mathrm{pt}}(\MM)$
has the following objects and morphisms: An object is a pair 
$\P^{(r_1,r_2)}= (\P, r_1,r_2 : I[2] \rightrightarrows U(\P) )$ 
consisting of a single-colored operad $\P\in \Op_{\{\ast\}}(\MM)$ and a parallel pair of
$\Seq_{\{\ast\}}(\MM)$-morphisms $r_1,r_2 : I[2] \rightrightarrows U(\P)$ (see \eqref{eqn:I2} for the definition
of $I[2]$). A morphism $\phi : \P^{(r_1,r_2)}\to \Q^{(s_1,s_2)}$ is an
$\Op_{\{\ast\}}(\MM)$-morphism $\phi : \P\to\Q$ that preserves the points, i.e.\
the diagram
\begin{flalign}
\xymatrix@C=4em{
\ar@{=}[d] I[2] \ar[r]^-{r_i} & U(\P) \ar[d]^-{U(\phi)}\\
I[2] \ar[r]_-{s_i} & U(\Q)
}
\end{flalign}
in $\Seq_{\{\ast\}}(\MM)$ commutes for $i=1,2$.
\end{defi}

\begin{propo}\label{prop:FToperadfunctoriality}
The assignment $(\ovr{\CC},\P^{(r_1,r_2)}) \longmapsto (\CC_0, \P^{(r_1,r_2)}_{\ovr{\CC}})$
of the field theory operads from Definition \ref{def:FToperad}
naturally extends to a functor $\OCat\times \Op_{\{\ast\}}^{2\mathrm{pt}}(\MM) \to 
\Op(\MM)$ with values in the category of operads with varying colors (see Definition \ref{def:varycolorOp}).
\end{propo}
\begin{proof}
For every morphism $(F,\phi) : (\ovr{\CC},\P^{(r_1,r_2)})\to (\ovr{\DD},\Q^{(s_1,s_2)})$ 
in $\OCat\times \Op_{\{\ast\}}^{2\mathrm{pt}}(\MM)$ one can define an $\Op(\MM)$-morphism
$\phi_F^{} : \P_\CC\to \Q_\DD$ between the corresponding auxiliary operads
from Definition \ref{def:PCoperad}. Concretely, this morphism
is specified by the components
\begin{flalign}
\xymatrix@C=3.5em{
\P_\CC\big(\substack{t \\ \und{c}}\big) \ar[r]^-{\phi_F^{}} & \Q_\DD\big(\substack{F(t) \\ F(\und{c})}\big)\\
\ar[u]^-{\iota_{\und{f}}} \P(n) \ar[r]_-{\phi}& \Q(n) \ar[u]_-{\iota_{F(\und{f})}}
}
\end{flalign}
We now show that the assignment of the field theory operads is functorial too. 
For this we first note that one can define analogously 
to above a morphism $R_{\perp_\CC} \to R_{\perp_\DD}$ of colored sequences
and one easily checks that this defines a morphism of parallel pairs in \eqref{eqn:Rperprelations}.
(For this step one uses that $F$ is an orthogonal functor and that $\phi$ preserves the points.)
Because forming colimits is functorial, this defines an $\Op_{\CC_0}(\MM)$-morphism 
$\P^{(r_1,r_2)}_{\ovr{\CC}}\to \colim\big(F^\ast(F(R_{\perp_{\DD}})) \rightrightarrows F^\ast(\Q_\DD)\big)$
to the coequalizer of the corresponding pullback operads. 
(With an abuse of notation, we denoted by $F$ both the free 
$\DD_0$-colored operad functor \eqref{eqn:freeforgetOp} and 
the orthogonal functor $F :\ovr{\CC}\to\ovr{\DD}$.)
Notice that pullback operads arise at this point because 
Definition \ref{def:FToperad} considers colimits in the categories of operads
with a {\em fixed} set of colors and not in the category $\Op(\MM)$.
From the universal property of colimits one obtains a canonical $\Op_{\CC_0}(\MM)$-morphism 
$\colim\big(F^\ast(F(R_{\perp_{\DD}})) \rightrightarrows F^\ast(\Q_\DD)\big) \to F^\ast(\Q^{(s_1,s_2)}_{\ovr{\DD}})$
to the pullback of field theory operad. The composition of the latter two morphisms
defines our desired $\Op(\MM)$-morphism, which we denote with abuse of notation
by the same symbol $\phi_F^{} : \P^{(r_1,r_2)}_{\ovr{\CC}}\to \Q^{(s_1,s_2)}_{\ovr{\DD}}$
as the one for the auxiliary operads.
\end{proof}

As a consequence of this proposition,
we obtain for every morphism 
$(F,\phi) : (\ovr{\CC},\P^{(r_1,r_2)})\to (\ovr{\DD},\Q^{(s_1,s_2)})$ in
$\OCat\times \Op_{\{\ast\}}^{2\mathrm{pt}}(\MM)$  an $\Op(\MM)$-morphism 
$\phi_F^{} : \P^{(r_1,r_2)}_{\ovr{\CC}} \to \Q^{(s_1,s_2)}_{\ovr{\DD}}$ and hence
by Theorems \ref{theo:adjunction} and \ref{theo:FTcatiso} an adjunction
\begin{flalign}\label{eqn:FTadjunction}
\xymatrix{
(\phi_F^{})_!^{} \,:\, \FT\big(\ovr{\CC},\P^{(r_1,r_2)}\big) ~\ar@<0.5ex>[r]&\ar@<0.5ex>[l]  ~\FT\big(\ovr{\DD},\Q^{(s_1,s_2)}\big) \,:\,
(\phi_F^{})_{}^\ast
}
\end{flalign}
between the corresponding categories of field theories. 
From the concrete definition of $\phi_F^{}$ given in the proof of
Proposition \ref{prop:FToperadfunctoriality} and the identification in Theorem \ref{theo:FTcatiso}, one observes that
the right adjoint $(\phi_F^{})^\ast$ admits a very explicit
description in terms of either of the two compositions in the commutative diagram
\begin{flalign}\label{eqn:FTadjunctionexplicit}
\xymatrix@C=4em{
  \ar[d]_-{(\phi^\ast)_\ast} \FT\big(\ovr{\CC},\Q^{(s_1,s_2)}\big) ~&~ \ar[l]_-{F^\ast} \FT\big(\ovr{\DD},\Q^{(s_1,s_2)}\big) 
  \ar[d]^-{(\phi^\ast)_\ast} \ar[dl]_-{(\phi_F^{})^\ast}   \\
 \FT\big(\ovr{\CC},\P^{(r_1,r_2)}\big)~&~  \ar[l]^-{F^\ast}\FT\big(\ovr{\DD},\P^{(r_1,r_2)}\big) .
 }
\end{flalign}
In this diagram $F^\ast$ is the restriction to the categories of field theories
of the pullback functor for functor categories
\begin{flalign}
F^\ast := (-)\circ F \,:\,  {\Alg_\O}^\DD ~\longrightarrow~ {\Alg_\O}^\CC\quad,
\end{flalign}
for $\O=\P$ and $\O=\Q$,  and $(\phi^\ast)_\ast$ is the restriction
to the categories of field theories of the pushforward functor for functor categories
\begin{flalign}
(\phi^\ast)_\ast := \phi^\ast \circ (-) \,:\,  {\Alg_{\Q}}^\EE~\longrightarrow {\Alg_{\P}}^\EE\quad,
\end{flalign}
for $\EE=\CC$ and $\EE=\DD$, where $\phi^\ast : \Alg_{\Q}\to \Alg_{\P}$
is the pullback functor corresponding to the single-colored operad morphism $\phi : \P\to\Q$.


\section{\label{sec:adjunctions}Universal constructions for field theories}
This section is concerned with analyzing in more depth the adjunctions in \eqref{eqn:FTadjunction}
and their relevance for universal constructions in field theory.
Because of \eqref{eqn:FTadjunctionexplicit}, this problem may be decomposed into
three smaller building blocks:
\begin{enumerate}
\item adjunctions induced by orthogonal functors $F : \ovr{\CC} \to \ovr{\DD}$
\begin{flalign}
\xymatrix{
F_! \,:\, \FT\big(\ovr{\CC},\P^{(r_1,r_2)}\big) ~\ar@<0.5ex>[r]&\ar@<0.5ex>[l]  ~\FT\big(\ovr{\DD},\P^{(r_1,r_2)}\big) \,:\,F^\ast
}
\end{flalign} 
\item adjunctions induced by $\Op_{\{\ast\}}^{2\mathrm{pt}}(\MM)$-morphisms $\phi : \P^{(r_1,r_2)}\to \Q^{(s_1,s_2)}$
\begin{flalign}
\xymatrix{
(\phi^\ast)^! \,:\, \FT\big(\ovr{\CC},\P^{(r_1,r_2)}\big) ~\ar@<0.5ex>[r]&\ar@<0.5ex>[l]  ~\FT\big(\ovr{\CC},\Q^{(s_1,s_2)}\big) \,:\,(\phi^\ast)_\ast
}
\end{flalign}
\item the interplay between these two cases via the diagram of categories and functors
\begin{flalign}\label{eqn:generalsquareadjunctions}
\xymatrix@C=4em{
  \ar@<0.5ex>[d]^-{(\phi^\ast)_\ast} \FT\big(\ovr{\CC},\Q^{(s_1,s_2)}\big) \ar@<0.5ex>[r]^-{F_!}~&~ \ar@<0.5ex>[l]^-{F^\ast} \FT\big(\ovr{\DD},\Q^{(s_1,s_2)}\big) \ar@<0.5ex>[d]^-{(\phi^\ast)_\ast}   \\
\ar@<0.5ex>[u]^-{(\phi^\ast)^!} \FT\big(\ovr{\CC},\P^{(r_1,r_2)}\big) \ar@<0.5ex>[r]^-{F_!}~&~  \ar@<0.5ex>[l]^-{F^\ast}\FT\big(\ovr{\DD},\P^{(r_1,r_2)}\big) \ar@<0.5ex>[u]^-{(\phi^\ast)^!}
 }
\end{flalign}
in which the square formed by the right adjoints commutes by \eqref{eqn:FTadjunctionexplicit}
and, as a consequence of the uniqueness (up to a unique natural isomorphism) of left adjoint functors, 
the square formed by the left adjoints commutes up to a unique natural isomorphism.
\end{enumerate}

In the following subsections we study particular classes of examples of such adjunctions, 
all of which are motivated by concrete problems and constructions in field theory, and discuss their
interplay. A particularly interesting example, which we will discuss later in Section \ref{sec:quantization}, 
is given by an adjunction that describes the quantization of linear field theories.

\subsection{Full orthogonal subcategories}
Recall from \cite{BeniniSchenkelWoike}
that a {\em full orthogonal subcategory} of an orthogonal
category $\ovr{\DD} = (\DD,\perp_\DD)$
is a full subcategory $\CC\subseteq \DD$
that is endowed with the pullback orthogonality relation, 
i.e.\ $f_1\perp_{\CC} f_2$ if and only if
$f_1\perp_\DD f_2$. The embedding functor $j : \CC\to \DD$
defines an orthogonal functor $j : \ovr{\CC} \to \ovr{\DD}$.
\begin{propo}\label{prop:extension}
Let $j : \ovr{\CC}\to \ovr{\DD}$ be a full orthogonal subcategory
and $\P^{(r_1,r_2)} \in \Op_{\{\ast\}}^{2\mathrm{pt}}(\MM)$ a bipointed single-colored
operad. Then the corresponding adjunction
\begin{flalign}\label{eqn:jadjunction}
\xymatrix{
j_! \,:\, \FT\big(\ovr{\CC},\P^{(r_1,r_2)}\big) ~\ar@<0.5ex>[r]&\ar@<0.5ex>[l]  ~\FT\big(\ovr{\DD},\P^{(r_1,r_2)}\big) \,:\,j^\ast
}
\end{flalign}
exhibits $\FT\big(\ovr{\CC},\P^{(r_1,r_2)}\big)$ as a full coreflective subcategory of 
$\FT\big(\ovr{\DD},\P^{(r_1,r_2)}\big)$, i.e.\ the unit $\eta : \id \to j^\ast\,j_! $
of this adjunction is a natural isomorphism.
\end{propo}
\begin{proof}
The proof is analogous to the corresponding one in \cite{BeniniSchenkelWoike}
and will not be repeated.
\end{proof}

\begin{ex}\label{ex:Locdiamond}
Recall the orthogonal category $\ovr{\Loc}$ of globally hyperbolic spacetimes 
from Example \ref{ex:Loc}. Consider the full subcategory $\Locc\subseteq \Loc$
of all spacetimes whose underlying manifold is diffeomorphic to $\bbR^m$. 
Endowed with the pullback orthogonality relation, i.e.\
$f_1\perp_{\Locc} f_2$ if and only if $f_1\perp_\Loc f_2$,
this defines a full orthogonal subcategory $ j : \ovr{\Locc} \to \ovr{\Loc}$.
The  corresponding adjunction is 
\begin{flalign}
\xymatrix{
j_! \,:\, \FT\big(\ovr{\Locc},\P^{(r_1,r_2)}\big) ~\ar@<0.5ex>[r]&\ar@<0.5ex>[l]  ~\FT\big(\ovr{\Loc},\P^{(r_1,r_2)}\big) \,:\,j^\ast
}\quad.
\end{flalign}
The right adjoint $j^\ast$ is the {\em restriction functor} which restricts field theories that are defined
on all of  $\ovr{\Loc}$ to the full orthogonal subcategory $\ovr{\Locc}$ of spacetimes diffeomorphic
to $\bbR^m$. More interestingly, the left adjoint $j_!$ is a {\em universal extension functor}
which extends field theories that are only defined on $\ovr{\Locc}$ to all of $\ovr{\Loc}$.
It was shown in \cite{BeniniSchenkelWoike} that $j_!$ is a generalization
and refinement of Fredenhagen's universal algebra construction \cite{Fre1,Fre2,Fre3,Lang}.
\sk

A non-trivial application of a similar universal extension functor for 
quantum field theories on spacetimes with boundaries has been studied 
in \cite{BeniniDappiaggiSchenkel}. It has been shown that the ideals of
the universal extension $j_!(\BBB)$ of a theory $\BBB$ that is defined 
only on the interior of a spacetime with boundary are related to
boundary conditions.
\end{ex}

\begin{rem}
The result in Proposition \ref{prop:extension} that $j_!$ exhibits
$\FT\big(\ovr{\CC},\P^{(r_1,r_2)}\big)$ as a full coreflective subcategory of 
$\FT\big(\ovr{\DD},\P^{(r_1,r_2)}\big)$ is crucial for a proper interpretation of
$j_!$ as a universal extension functor and $j^\ast$ as a restriction functor
in the spirit of Example \ref{ex:Locdiamond}. Given any field theory 
$\BBB \in \FT\big(\ovr{\CC},\P^{(r_1,r_2)}\big)$ on the full orthogonal subcategory
$\ovr{\CC}\subseteq \ovr{\DD}$, one may apply the left and then the right adjoint
functor in \eqref{eqn:jadjunction} to obtain another field theory
$j^\ast j_!(\BBB)\in  \FT\big(\ovr{\CC},\P^{(r_1,r_2)}\big)$ on $\ovr{\CC}\subseteq \ovr{\DD}$.
The latter is interpreted as the restriction of the universal extension of $\BBB$.
By Proposition  \ref{prop:extension}, the unit $\eta_\BBB^{} :\BBB\to j^\ast j_!(\BBB)$
defines an isomorphism between these two theories, which means that $j_!$ 
extends field theories from $\ovr{\CC}\subseteq \ovr{\DD}$ to all of $\ovr{\DD}$
without altering their values on the subcategory $\ovr{\CC}$. 
\sk

With this observation in mind, we would like to comment on existing
criticisms that universal constructions, such as Fredenhagen's universal algebra
or our universal extension $j_!$, may fail to provide a non-trivial result, 
see e.g.\ \cite{Ruzzi} and also \cite{Lang}. (We would like to thank the anonymous
referee for bringing this to our attention.)
It is indeed true that the algebra $j_!(\BBB)(d)\in \Alg_\P$ associated
to an object $d\in\DD$ that is not in the subcategory $\CC\subseteq\DD$, i.e.\ $d\not\in\CC$,
might be trivial. However, for every non-trivial $\BBB\in\FT\big(\ovr{\CC},\P^{(r_1,r_2)}\big)$, the universally 
extended field theory $j_!(\BBB)\in \FT\big(\ovr{\DD},\P^{(r_1,r_2)}\big)$ as a whole
is non-trivial because, as we explained above, its restriction $j^\ast j_!(\BBB)$
to $\ovr{\CC}$ is isomorphic to the input $\BBB$ of the construction. 
We expect that one can construct physical examples of such theories that are 
non-trivial on simple spacetimes in $\ovr{\CC}$, but might be trivial on certain 
complicated spacetimes in $\ovr{\DD}$, by considering field equations 
that admit only local solutions.
\end{rem}

An interesting application of the class of adjunctions in \eqref{eqn:jadjunction} 
is that they allow us to formalize a kind of local-to-global (i.e.\ {\em descent}) 
condition for field theories. Given a field theory $\AAA \in \FT\big(\ovr{\DD},\P^{(r_1,r_2)}\big)$
on the bigger category $\ovr{\DD}$, one may ask whether it is already completely determined
by its values on the full orthogonal subcategory $\ovr{\CC}\subseteq\ovr{\DD}$. 
In the context of Example \ref{ex:Locdiamond}, this means asking if
the value of a field theory on a general spacetime $M\in\ovr{\Loc}$ is completely determined 
by its behavior on spacetimes diffeomorphic to $\bbR^m$, which is a typical question of descent.
The following definition provides a formalization of this idea.
\begin{defi}\label{def:jlocal}
A field theory $\AAA \in \FT\big(\ovr{\DD},\P^{(r_1,r_2)}\big)$ on $\ovr{\DD}$
is called {\em $j$-local} if the corresponding component of the counit
$\epsilon_\AAA^{} : j_! \,j^\ast(\AAA)\to \AAA$ is an isomorphism.
The full subcategory of $j$-local field theories is denoted by
$\FT\big(\ovr{\DD},\P^{(r_1,r_2)}\big)^{j\mathrm{{-}loc}}\subseteq \,\FT\big(\ovr{\DD},\P^{(r_1,r_2)}\big)$.
\end{defi}

The following result, which extends earlier results from \cite{BeniniSchenkelWoike} to
our more general framework, shows that $j$-local field theories on the bigger category 
$\ovr{\DD}$  may be equivalently described by field theories on the full orthogonal
subcategory $\ovr{\CC}\subseteq \ovr{\DD}$. 
\begin{cor}
The adjunction \eqref{eqn:jadjunction} restricts to an adjoint equivalence
\begin{flalign}
\xymatrix{
j_! \,:\, \FT\big(\ovr{\CC},\P^{(r_1,r_2)}\big) ~\ar@<0.8ex>[r]_-{\sim}&\ar@<0.8ex>[l]  ~\FT\big(\ovr{\DD},\P^{(r_1,r_2)}\big)^{j\mathrm{{-}loc}} \,:\,j^\ast
}\quad.
\end{flalign}
\end{cor}
\begin{proof}
This is an immediate consequence of Proposition \ref{prop:extension}.
\end{proof}

\begin{ex}
Being a powerful local-to-global property, it is in general not easy to prove that 
a given field theory $\AAA \in \FT\big(\ovr{\DD},\P^{(r_1,r_2)}\big)$ on $\ovr{\DD}$
is $j$-local for some full orthogonal subcategory embedding $j:\ovr{\CC}\to\ovr{\DD}$.
Positive results are known for the usual Klein-Gordon quantum field theory and
$j:\ovr{\Locc}\to\ovr{\Loc}$, see \cite{Lang} and \cite[Section 5]{BeniniSchenkelWoike}.
We expect that this proof can be adapted to cover all vector bundle valued linear 
quantum field theories in the sense of \cite{Baer,BGproc,BeniniDappiaggiHack}.
\end{ex}

\subsection{\label{subsec:orthlocalication}Orthogonal localizations}
Recall from \cite{BeniniSchenkelWoike} that
the {\em orthogonal localization} of an orthogonal category 
$\ovr{\CC}$ at a subset $W \subseteq \mathrm{Mor}\,\CC$ of the set of 
morphisms is given by the localized category $\CC[W^{-1}]$ endowed with
the pushforward orthogonality relation $\perp_{\CC[W^{-1}]} \,:= L_\ast(\perp_\CC)$
along the localization functor $L : \CC\to\CC[W^{-1}]$,
i.e.\ $\perp_{\CC[W^{-1}]}$ is the smallest orthogonality relation 
such that $L(f_1)\perp_{\CC[W^{-1}]} L(f_2)$ for all $f_1\perp_\CC f_2$. 
The localization functor defines an orthogonal functor $L : \ovr{\CC}\to \ovr{\CC[W^{-1}]}$.
\begin{propo}\label{prop:timeslice}
Let $L: \ovr{\CC}\to \ovr{\CC[W^{-1}]}$ be an orthogonal localization
and $\P^{(r_1,r_2)} \in \Op_{\{\ast\}}^{2\mathrm{pt}}(\MM)$ a
bipointed single-colored operad. Then the corresponding adjunction
\begin{flalign}\label{eqn:Ladjunction}
\xymatrix{
L_! \,:\, \FT\big(\ovr{\CC},\P^{(r_1,r_2)}\big) ~\ar@<0.5ex>[r]&\ar@<0.5ex>[l]  ~\FT\big(\ovr{\CC[W^{-1}]},\P^{(r_1,r_2)}\big) \,:\,L^\ast
}
\end{flalign}
exhibits $\FT\big(\ovr{\CC[W^{-1}]},\P^{(r_1,r_2)}\big) $ as a full reflective subcategory of 
$\FT\big(\ovr{\CC},\P^{(r_1,r_2)}\big)$, i.e.\ the counit $\epsilon: L_!\, L^\ast\to \id$
of this adjunction is a natural isomorphism.
\end{propo}
\begin{proof}
The proof is analogous to the corresponding one in \cite{BeniniSchenkelWoike}
and will not be repeated.
\end{proof}

Let us now explain in some detail the relationship between 
the adjunction \eqref{eqn:Ladjunction} and a suitable generalization of
the `time-slice axiom' that we shall call $W$-constancy.
\begin{defi}\label{def:Wconstant}
A field theory $\AAA\in \FT\big(\ovr{\CC},\P^{(r_1,r_2)}\big)$ is called {\em $W$-constant}
if the $\Alg_\P$-morphism $\AAA(f) : \AAA(c)\to\AAA(c^\prime)$ is an isomorphism,
for all $(f:c\to c^\prime)\in W$.
The full subcategory of $W$-constant field theories
is denoted by $\FT\big(\ovr{\CC},\P^{(r_1,r_2)}\big)^{W\mathrm{{-}const}}
\subseteq \FT\big(\ovr{\CC},\P^{(r_1,r_2)}\big)$.
\end{defi}
\begin{propo}\label{prop:Ladjunctionequivalence}
The adjunction \eqref{eqn:Ladjunction} restricts to an adjoint equivalence
\begin{flalign}\label{eqn:Ladjunctionequivalence}
\xymatrix{
L_! \,:\, \FT\big(\ovr{\CC},\P^{(r_1,r_2)}\big)^{W\mathrm{{-}const}} ~\ar@<0.8ex>[r]_-{\sim}&\ar@<0.8ex>[l]  ~\FT\big(\ovr{\CC[W^{-1}]},\P^{(r_1,r_2)}\big) \,:\,L^\ast
}\quad.
\end{flalign}
As a consequence, a field theory $\AAA\in \FT\big(\ovr{\CC},\P^{(r_1,r_2)}\big)$ is $W$-constant
if and only if the corresponding component $\eta_\AAA^{} : \AAA\to L^\ast L_!(\AAA)$
of the unit of the adjunction \eqref{eqn:Ladjunction} is an isomorphism.
\end{propo}
\begin{proof}
By Proposition \ref{prop:timeslice}, we already know that 
the right adjoint functor $L^\ast$ is fully faithful, hence it remains
to prove that its essential image is $\FT\big(\ovr{\CC},\P^{(r_1,r_2)}\big)^{W\mathrm{{-}const}} $.
The image of $L^\ast$ lies in $\FT\big(\ovr{\CC},\P^{(r_1,r_2)}\big)^{W\mathrm{{-}const}} $
because, for every $\BBB\in \FT\big(\ovr{\CC[W^{-1}]},\P^{(r_1,r_2)}\big)$, the field theory
$L^\ast(\BBB)$ is $W$-constant since $L^\ast = (-)\circ L$ is given by restricting 
the pullback functor for functor categories and the localization functor $L:\CC\to\CC[W^{-1}]$ 
maps morphisms in $W$ to isomorphisms. To prove essential surjectivity, let
$\AAA \in \FT\big(\ovr{\CC},\P^{(r_1,r_2)}\big)^{W\mathrm{{-}const}} $
and consider its underlying functor $\AAA :\CC\to \Alg_\P$.
By definition of localization, there exists a functor $\BBB : \CC[W^{-1}]\to \Alg_\P$
and a natural isomorphism $\AAA\cong L^\ast(\BBB)$. Using that the orthogonality
relation $\perp_{\CC[W^{-1}]}$ is generated by $L(f_1)\perp_{\CC[W^{-1}]} L(f_2)$, for
all $f_1\perp_\CC f_2$, one easily checks that $\BBB\in \FT\big(\ovr{\CC[W^{-1}]},\P^{(r_1,r_2)}\big)$.
\end{proof}

\begin{ex}\label{ex:timeslice}
Recall the orthogonal category $\ovr{\Loc}$ of globally hyperbolic spacetimes 
from Example \ref{ex:Loc} and consider the subset  
$W \subseteq \mathrm{Mor}\,\Loc$ of all Cauchy morphisms,
i.e.\ morphisms $f:M\to M^\prime$ whose image $f(M)\subseteq M^\prime$ contains
a Cauchy surface of $M^\prime$.
In this case $W$-constant field theories are precisely field theories that 
satisfy the usual time-slice axiom with respect to all Cauchy morphisms. 
As a consequence of Proposition \ref{prop:Ladjunctionequivalence}, 
such field theories can be described equivalently as field theories on
the orthogonal localization $\ovr{\Loc[W^{-1}]}$ of $\ovr{\Loc}$ 
at all Cauchy morphisms $W$. This alternative point of view
comes together with an adjunction
\begin{flalign}\label{eqn:LadjunctionLoc}
\xymatrix{
L_! \,:\, \FT\big(\ovr{\Loc},\P^{(r_1,r_2)}\big) ~\ar@<0.5ex>[r]&\ar@<0.5ex>[l]  ~\FT\big(\ovr{\Loc[W^{-1}]},\P^{(r_1,r_2)}\big) \,:\,L^\ast
}\quad,
\end{flalign}
which allows us to detect $W$-constancy of a field theory $\AAA\in \FT\big(\ovr{\Loc},\P^{(r_1,r_2)}\big)$
by testing if the unit $\eta_\AAA^{} : \AAA\to L^\ast\,L_!(\AAA)$ is an isomorphism.
\sk

The right adjoint $L^\ast$ of the adjunction \eqref{eqn:LadjunctionLoc} can be
interpreted as the functor that forgets that a field theory 
$\BBB\in \FT\big(\ovr{\Loc[W^{-1}]},\P^{(r_1,r_2)}\big)$ satisfies the time-slice axiom.
More interestingly, the left adjoint $L_!$ assigns to a field theory $\AAA \in \FT\big(\ovr{\Loc},\P^{(r_1,r_2)}\big)$
that does {\em not} necessarily satisfy the time-slice axiom a theory that does. Hence,
one may call the left adjoint functor $L_!$ a `time-slicification' functor.
Notice that the result in Proposition \ref{prop:timeslice} that $L^\ast$ exhibits 
$ \FT\big(\ovr{\Loc[W^{-1}]},\P^{(r_1,r_2)}\big)$ as a full reflective subcategory of
$\FT\big(\ovr{\Loc},\P^{(r_1,r_2)}\big)$ has a concrete meaning. The isomorphisms
$\epsilon_{\BBB}^{} : L_!\, L^\ast(\BBB)\to\BBB$ given by the counit say that time-slicification
does not alter those field theories that already do satisfy the time-slice axiom, which is 
of course a very reasonable property.
\sk

To conclude we consider the following example
in order to show that our `time-slicification' functor does not
generically produce trivial field theories. 
Let $\BBB= \AAA/\mathcal{I} \in \FT\big(\ovr{\Loc[W^{-1}]},\P^{(r_1,r_2)}\big)$
be a field theory satisfying the time-slice axiom 
that is obtained by quotienting out 
an equation of motion ideal $\mathcal{I}$ 
of an `off-shell' field theory $\AAA$.
More formally, this means that $L^\ast(\BBB)\in \FT\big(\ovr{\Loc},\P^{(r_1,r_2)}\big)$
is given by a coequalizer
\begin{flalign}
\xymatrix@C=3em{
F(\mathcal{I}) \,\ar@<0.5ex>[r]\ar@<-0.5ex>[r] \,& \,\AAA \ar@{-->}[r]\,&\, L^\ast(\BBB)
}
\end{flalign}
in the category $\FT\big(\ovr{\Loc},\P^{(r_1,r_2)}\big)$,
where $F(\mathcal{I})$ is the freely-generated field theory of the equation of motion
ideal $\mathcal{I}$ and the two maps are given extending via
the free-forget adjunction $F\dashv U$ the
inclusion $\mathcal{I} \hookrightarrow U(\AAA)$ and the
zero map $0:\mathcal{I}\to U(\AAA)$ to $F(\mathcal{I})$.
The `off-shell' field theory $\AAA\in \FT\big(\ovr{\Loc},\P^{(r_1,r_2)}\big)$ 
is of course not assumed to satisfy the time-slice axiom. 
Because the `time-slicification' functor $L_!$ is left adjoint
it preserves colimits and hence
\begin{flalign}
\xymatrix@C=3em{
L_! F(\mathcal{I}) \,\ar@<0.5ex>[r]\ar@<-0.5ex>[r] \,& \,L_!(\AAA) \ar@{-->}[r]\,&\, L_! L^\ast(\BBB)\ar[r]^-{\epsilon_\BBB^{}}_-{\cong}\,&\,\BBB
}
\end{flalign}
is a coequalizer in $\FT\big(\ovr{\Loc[W^{-1}]},\P^{(r_1,r_2)}\big)$,
where we used Proposition \ref{prop:timeslice} for the last arrow.
We see that our field theory $\BBB=\AAA/\mathcal{I}$ may be equivalently presented 
as a quotient of the `time-slicification' $L_!(\AAA)$ of the `off-shell' field theory $\AAA$. 
In other words, every `on-shell' quotient of $\AAA$ can be presented as a quotient of $L_! (\AAA)$,
hence $L_!(\AAA)$ must be non-trivial provided that $\AAA$ admits non-trivial `on-shell' quotients.
\end{ex}

\subsection{Change of bipointed single-colored operad}
Our third class of examples are adjunctions that correspond to 
morphisms $\phi : \P^{(r_1,r_2)}\to \Q^{(s_1,s_2)}$ of bipointed single-colored
operads, i.e.\
\begin{flalign}\label{eqn:phiadjunction}
\xymatrix{
(\phi^\ast)^! \,:\, \FT\big(\ovr{\CC},\P^{(r_1,r_2)}\big) ~\ar@<0.5ex>[r]&\ar@<0.5ex>[l]  ~\FT\big(\ovr{\CC},\Q^{(s_1,s_2)}\big) \,:\,(\phi^\ast)_\ast
}\quad.
\end{flalign}
Let us stress that these adjunctions are conceptually very different
to the ones we studied in the previous two subsections because they change the
type of field theories and not the orthogonal category on which field theories
are defined. In particular, such adjunctions can {\em not} be formulated within
the original operadic framework for algebraic quantum field theory 
developed in \cite{BeniniSchenkelWoike} as they crucially rely on the more flexible 
definition \ref{def:FToperad} of field theory operads.
In Section \ref{sec:quantization} we study
an interesting example given by an adjunction that describes 
the quantization of linear field theories. 
\sk
 
We observe the following preservation results for $j$-local field theories 
(see Definition \ref{def:jlocal}) and for $W$-constant field theories 
(see Definition \ref{def:Wconstant}) under the adjunctions \eqref{eqn:phiadjunction}.
\begin{propo}\label{prop:preservation}
Let $\phi : \P^{(r_1,r_2)}\to \Q^{(s_1,s_2)}$ be an $\Op_{\{\ast\}}^{2\mathrm{pt}}(\MM)$-morphism,
$j : \ovr{\CC}\to\ovr{\DD}$ a full orthogonal subcategory and $W\subseteq \mathrm{Mor}\,\CC$
a subset.
\begin{itemize}
\item[a)] The left adjoint functor 
$(\phi^\ast)^! : \FT\big(\ovr{\DD},\P^{(r_1,r_2)}\big) \to \FT\big(\ovr{\DD},\Q^{(s_1,s_2)}\big)$
preserves $j$-local field theories.

\item[b)] The right adjoint functor $(\phi^\ast)_\ast : \FT\big(\ovr{\CC},\Q^{(s_1,s_2)}\big) 
\to \FT\big(\ovr{\CC},\P^{(r_1,r_2)}\big)$ preserves $W$-constant field theories.
\end{itemize} 
\end{propo}
\begin{proof}
{\em Item a):} Let $\AAA\in \FT\big(\ovr{\DD},\P^{(r_1,r_2)}\big)^{j\mathrm{{-}loc}}$ 
be any $j$-local field theory of type $\P^{(r_1,r_2)}$,
i.e.\  $\epsilon_\AAA^{} : j_!\,j^\ast(\AAA)\to \AAA$ is an isomorphism. The claim is that the field theory 
$(\phi^\ast)^!(\AAA)\in \FT\big(\ovr{\DD},\Q^{(s_1,s_2)}\big)$
of type $\Q^{(s_1,s_2)}$ is $j$-local as well, i.e.\ 
$\epsilon_{(\phi^\ast)^!(\AAA)}^{} : j_!\,j^\ast(\phi^\ast)^!(\AAA)\to (\phi^\ast)^!(\AAA)$
is an isomorphism.  This follows from the commutative diagram
\begin{flalign}\label{eqn:jpreservation}
\xymatrix@C=3em{
j_!\,j^\ast\,(\phi^\ast)^!(\AAA) \ar[rr]^-{\epsilon_{(\phi^\ast)^!(\AAA)}^{}} && (\phi^\ast)^!(\AAA)\\
\ar[u]_-{\cong}^-{j_!\,j^\ast\,(\phi^\ast)^!\epsilon_\AAA^{}} j_!\,j^\ast\,(\phi^\ast)^!\, j_!\, j^\ast(\AAA) \ar[rr]^-{\epsilon_{(\phi^\ast)^!\,j_!\, j^\ast(\AAA)}^{}} && (\phi^\ast)^!\, j_!\,j^\ast(\AAA) \ar[u]^-{\cong}_-{(\phi^\ast)^! \epsilon_\AAA^{}}\\
\ar[u]^-{\cong} j_!\,j^\ast\, j_! \, (\phi^\ast)^!\, j^\ast(\AAA) \ar[rr]^-{\epsilon_{j_!\,(\phi^\ast)^!\, j^\ast(\AAA)}^{}} && j_!\,(\phi^\ast)^!\, j^\ast(\AAA) \ar[u]_-{\cong}\\
\ar[u]_\cong^-{j_! \eta_{(\phi^\ast)^!\, j^\ast(\AAA)}} j_!\,(\phi^\ast)^!\, j^\ast(\AAA)\ar@{=}[rru]&&
}
\end{flalign}
where isomorphisms are indicated by $\cong$. In more detail, the top square commutes by naturality of
the counit and the vertical arrows are isomorphisms because $\AAA$ is $j$-local. The middle square
commutes because of \eqref{eqn:generalsquareadjunctions}. 
The bottom triangle is the triangle identity for the adjunction and the unit (vertical arrow) is an isomorphism
because of Proposition \ref{prop:extension}.
\sk

{\em Item b):} This is immediate because $(\phi^\ast)_\ast = \phi^\ast\circ(-)$ is given by
restricting the pushforward functor for functor categories and every functor 
$\phi^\ast$ preserves isomorphisms.
\end{proof}

Let us emphasize that the result in Proposition \ref{prop:preservation} 
is asymmetric in the sense that $j$-local field theories
are preserved by the {\em left} adjoints $(\phi^\ast)^!$ and
$W$-constant field theories are preserved by the {\em right}
adjoints $(\phi^\ast)_\ast$. 
The opposite preservation properties do not hold true in general,
however we would like to note the following special case in which there exists
a further preservation result. This will become relevant in Section \ref{sec:quantization} below.
\begin{propo}\label{prop:specialpreservation}
Let $\phi : \P^{(r_1,r_2)}\to \Q^{(s_1,s_2)}$ be an $\Op_{\{\ast\}}^{2\mathrm{pt}}(\MM)$-morphism 
and $W\subseteq \mathrm{Mor}\,\CC$ a subset. Suppose that the left adjoint functor
$(\phi^\ast)^! : \FT(\ovr{\CC},\P^{(r_1,r_2)})\to \FT(\ovr{\CC},\Q^{(s_1,s_2)})$
is (naturally isomorphic to) the restriction to the categories of field theories of the pushforward
functor for functor categories
\begin{flalign}
(\phi_!)_\ast := \phi_! \circ (-)\, :\, {\Alg_\P}^{\CC} ~\longrightarrow ~{\Alg_\Q}^{\CC}\quad,
\end{flalign}
where the adjunction $\phi_! : \Alg_\P \rightleftarrows \Alg_\Q : \phi^\ast$
corresponds to the single-colored 
operad morphism $\phi : \P\to\Q$. Then the left adjoint functor 
$(\phi^\ast)^! : \FT(\ovr{\CC},\P^{(r_1,r_2)})\to \FT(\ovr{\CC},\Q^{(s_1,s_2)})$
preserves $W$-constant field theories.
\end{propo}
\begin{proof}
This is immediate because by hypothesis there is a natural isomorphism 
$(\phi^\ast)^! \cong \phi_!\circ(-)$ and every functor $\phi_!$ preserves isomorphisms.
\end{proof}

We conclude this section with a technical lemma that provides a criterion
to detect whether the hypotheses of Proposition \ref{prop:specialpreservation}
are fulfilled. Recall from Definition \ref{def:FToperad} that there
exists a natural projection $\Op_{\CC_0}(\MM)$-morphism
$\pi : \P_\CC\to \P^{(r_1,r_2)}_{\ovr{\CC}}$ from our auxiliary 
operads to the field theory operads. Given any $\Op_{\{\ast\}}^{2\mathrm{pt}}(\MM)$-morphism 
$\phi : \P^{(r_1,r_2)}\to \Q^{(s_1,s_2)}$, this yields the square of adjunctions
\begin{flalign}
\xymatrix@C=4em{
  \ar@<0.5ex>[d]^-{\pi^\ast} \FT\big(\ovr{\CC},\P^{(r_1,r_2)}\big) \ar@<0.5ex>[r]^-{(\phi^\ast)^!}~&~ \ar@<0.5ex>[l]^-{(\phi^\ast)_\ast} \FT\big(\ovr{\CC},\Q^{(s_1,s_2)}\big) \ar@<0.5ex>[d]^-{\pi^\ast}   \\
\ar@<0.5ex>[u]^-{\pi_!} {\Alg_\P}^\CC \ar@<0.5ex>[r]^-{(\phi_!)_\ast} ~&~  \ar@<0.5ex>[l]^-{(\phi^\ast)_\ast} {\Alg_\Q}^\CC \ar@<0.5ex>[u]^-{\pi_!}
 }
\end{flalign}
in which the square formed by the right adjoints commutes, i.e.\ 
$(\phi^\ast)_\ast~\pi^\ast \, =\,  \pi^\ast~(\phi^\ast)_\ast$,
and hence the square formed by the left adjoints commutes (up to a unique natural isomorphism), i.e.\
$(\phi^\ast)^!~\pi_!\,\cong\, \pi_!~(\phi_!)_\ast$. Notice that the vertical adjunctions
exhibit the field theory categories as full reflective subcategories of the functor categories.
An immediate consequence is the following
\begin{lem}\label{lem:detection}
If the functor $(\phi_!)_\ast\,\pi^\ast : \FT\big(\ovr{\CC},\P^{(r_1,r_2)}\big)  \to {\Alg_\Q}^\CC$
factors through the full reflective subcategory $\FT\big(\ovr{\CC},\Q^{(s_1,s_2)}\big) \subseteq {\Alg_\Q}^\CC$,
then the left adjoint  $(\phi^\ast)^! : \FT\big(\ovr{\CC},\P^{(r_1,r_2)}\big) \to \FT\big(\ovr{\CC},\Q^{(s_1,s_2)}\big)$
is (naturally isomorphic to) the restriction to the categories of field theories of 
the pushforward functor $(\phi_!)_\ast : {\Alg_\P}^\CC \to {\Alg_\Q}^\CC$.
\end{lem}


\section{\label{sec:quantization}Linear quantization adjunction}
Throughout this section we assume that the underlying bicomplete closed
symmetric monoidal category $\MM$ is additive.
Recalling Example \ref{ex:LCQFT} and also Remark \ref{rem:LCQFTviacommutator}, 
we define the category of {\em quantum field theories} on an orthogonal
category $\ovr{\CC}$ by
\begin{flalign}
\QFT(\ovr{\CC}) \,:=\, \FT\big(\ovr{\CC}, \mathsf{As}^{([\cdot,\cdot],0)}\big) \quad.
\end{flalign}
Recalling further Example \ref{ex:LinearFT}, we define the category of
{\em linear field theories} on $\ovr{\CC}$ by
\begin{flalign}
\LFT(\ovr{\CC})\,:=\, \FT\big(\ovr{\CC}, \mathsf{uLie}^{([\cdot,\cdot],0)}\big) \quad.
\end{flalign}
We define an $\Op_{\{\ast\}}(\MM)$-morphism 
\begin{flalign}\label{eqn:phiLieAss}
\phi \, :\,  \mathsf{uLie} ~\longrightarrow~ \mathsf{As}~,~~
\begin{cases} 
~~\eta ~~\longmapsto~ \eta~~,\\
[\cdot,\cdot] ~\longmapsto~ \mu - \mu^\op~~,
\end{cases}
\end{flalign}
which one can confirm is well-defined by using the relations of 
the associative operad (see Example \ref{ex:Ass}) and the ones
of the unital Lie operad (see Example \ref{ex:OpuLie}).
It is evident that $\phi  : \mathsf{uLie}^{([\cdot,\cdot],0)} \to\mathsf{As}^{([\cdot,\cdot],0)}$
defines an $\Op_{\{\ast\}}^{2\mathrm{pt}}(\MM)$-morphism
in the sense of Definition \ref{def:bipointed}. By \eqref{eqn:phiadjunction}
this induces an adjunction between the category of linear field theories
and the category of quantum field theories, which we shall denote by
\begin{flalign}\label{eqn:QUadjunction}
\xymatrix{
(\phi^\ast)^! = \Qlin \,:\, \LFT(\ovr{\CC}) ~\ar@<0.5ex>[r]&\ar@<0.5ex>[l]  ~\QFT(\ovr{\CC})  \,:\,  \Ulin = (\phi^\ast)_\ast
}\quad.
\end{flalign}
The aim of this section is to study this adjunction in detail and in particular to show that
the left adjoint $\Qlin$ admits an interpretation as a linear quantization functor.
\sk

Let us first provide an explicit description of the right adjoint functor
$\Ulin = (\phi^\ast)_\ast$. Note that the functor 
$\phi^\ast : \Alg_{\mathsf{As}} \to \Alg_{\mathsf{uLie}}$ from associative and unital 
algebras to unital Lie algebras is very explicit. It assigns to any
$(A,\mu_A,\eta_A)\in \Alg_{\mathsf{As}}$ the unital Lie algebra
$\phi^\ast(A,\mu_A,\eta_A)  = (A, \mu_A-\mu_A^\op ,\eta_A)\in \Alg_{\mathsf{uLie}}$, 
where the Lie bracket is given by the commutator. 
The corresponding pushforward functor 
$\Ulin = (\phi^\ast)_\ast : \QFT(\ovr{\CC}) \to \LFT(\ovr{\CC})$
carries out this construction object-wise on $\CC$. Concretely, for 
$\big(\AAA: \CC\to \Alg_{\mathsf{As}}\big) \in \QFT(\ovr{\CC})$,
the functor underlying $\Ulin(\AAA)\in\LFT(\ovr{\CC})$ is 
given by $\Ulin(\AAA)(c) = \phi^\ast\big(\AAA(c)\big) \in \Alg_{\mathsf{uLie}}$, 
for all $c\in \CC$.
\sk

We now provide an explicit description of the left adjoint functor $\Qlin$
in \eqref{eqn:QUadjunction}. Our strategy is to analyze the pushforward functor 
$(\phi_!)_\ast : {\Alg_{\mathsf{uLie}}}^\CC \to{\Alg_{\mathsf{As}}}^\CC$
for the functor categories and to prove that it satisfies the criterion of Lemma \ref{lem:detection}.
As a consequence of this lemma, the restriction to the categories of field theories
of the pushforward functor $(\phi_!)_\ast$ defines a model for the left adjoint functor $\Qlin$. 
\sk

Let us describe first the left adjoint functor
of the adjunction $\phi_! : \Alg_{\mathsf{uLie}} \rightleftarrows \Alg_{\mathsf{As}}:\phi^\ast$
between algebras over single-colored operads. The following construction, which 
we will call the {\em unital universal enveloping algebra construction}, defines a 
model for the left adjoint $\phi_! $.
Let $V\in\Alg_{\mathsf{uLie}}$ be any unital Lie algebra,
with Lie bracket  $[\cdot,\cdot] : V\otimes V\to V$  and unit $\eta : I\to V$.
As the first step, we form the usual tensor algebra 
$T^{\otimes} V := \bigoplus_{n=0}^{\infty} V^{\otimes n} \in\Alg_{\mathsf{As}}$,
i.e.\ the free $\mathsf{As}$-algebra of the underlying object $V\in\MM$,
with multiplication $\mu_\otimes : T^{\otimes} V \otimes T^{\otimes} V\to T^{\otimes} V$
and unit $\eta_\otimes : I \to T^{\otimes} V$. We then
consider the two parallel $\MM$-morphisms
\begin{subequations}\label{eqn:UEA}
\begin{flalign}
\xymatrix@C=4.5em{
V\otimes V \ar@<0.5ex>[rr]^-{q_1 \, \, :=\, \, (\mu_\otimes -\mu_\otimes^\op)~(\iota_1\otimes \iota_1)} 
 \ar@<-0.5ex>[rr]_-{ q_2 \,\,:=\,\,  \iota_1~[\cdot,\cdot]}~&&~ T^\otimes V 
}\quad,
\end{flalign}
where $\iota_1 : V\to T^{\otimes }V$ is the inclusion into the coproduct,
which compare the commutator of $T^{\otimes}V$ with the Lie bracket of $V$.
We form the corresponding coequalizer
\begin{flalign}
\xymatrix@C=3em{
T^\otimes(V\otimes V) \ar@<0.5ex>[r]^-{q_1} \ar@<-0.5ex>[r]_-{q_2}~&~ T^\otimes V 
\ar@{-->}[r]^-{\pi} ~&~ U^\otimes V
}
\end{flalign}
\end{subequations}
in $\Alg_{\mathsf{As}}$ and notice that $U^\otimes V$ is the universal enveloping algebra
of the underlying Lie algebra $(V,[\cdot,\cdot])\in\Alg_{\mathsf{Lie}}$. As the final step, 
we consider the two parallel $\MM$-morphisms
\begin{subequations}\label{eqn:uUEA}
\begin{flalign}
\xymatrix@C=2.5em{
I  \ar@<0.5ex>[rr]^-{ s_1\, \,:= \,\,\pi~\iota_1~\eta } \ar@<-0.5ex>[rr]_-{s_2\,\,:=\,\, \pi~\eta_\otimes }~&&~ U^\otimes V
}\quad,
\end{flalign}
which compare the unit of $V$ with the unit of $T^\otimes V$,
and form the corresponding coequalizer
\begin{flalign}
\xymatrix@C=3em{
T^\otimes(I) \ar@<0.5ex>[r]^-{s_1} \ar@<-0.5ex>[r]_-{s_2}~&~ U^\otimes V 
\ar@{-->}[r]^{\pi^\prime} ~&~ \phi_!(V)
}
\end{flalign}
\end{subequations}
in $\Alg_{\mathsf{As}}$. All of these constructions are clearly functorial.
\begin{lem}\label{lem:phi!LieAss}
The functor $\phi_! : \Alg_{\mathsf{uLie}}\to \Alg_{\mathsf{As}}$ described above 
is left adjoint to the functor $\phi^\ast : \Alg_{\mathsf{As}} \to \Alg_{\mathsf{uLie}}$.
\end{lem}
\begin{proof}
It is easy to construct a natural bijection
$\Hom_{\Alg_{\mathsf{As}}}(\phi_!(V),A)\cong \Hom_{\Alg_{\mathsf{uLie}}}(V,\phi^\ast(A))$,
for all $V\in \Alg_{\mathsf{uLie}}$ and $A\in\Alg_{\mathsf{As}}$. Concretely,
given $\kappa : \phi_!(V)\to A$ in $\Alg_{\mathsf{As}}$, 
then $\kappa \, \pi^\prime\,\pi\,\iota_1 : V\to\phi^\ast(A)$ defines
an $\Alg_{\mathsf{uLie}}$-morphism. On the other hand, given
$\rho : V\to \phi^\ast(A)$ in $\Alg_{\mathsf{uLie}}$, then the canonical 
extension to an $\Alg_{\mathsf{As}}$-morphism  $\rho : T^\otimes V\to A$ 
on the tensor algebra descends to the quotients in 
\eqref{eqn:UEA} and \eqref{eqn:uUEA}. 
\end{proof}
\begin{propo}
For every linear field theory $(\BBB : \CC\to \Alg_{\mathsf{uLie}})\in\LFT(\ovr{\CC})$,
the functor $(\phi_!)_\ast(\BBB) : \CC\to\Alg_{\mathsf{As}}$ 
defines a quantum field theory, i.e.\ $(\phi_!)_\ast(\BBB)\in\QFT(\ovr{\CC})$.
\end{propo}
\begin{proof}
By hypothesis, given any orthogonal pair $(f_1 :c_1\to c)\perp (f_2:c_2\to c)$ in $\ovr{\CC}$,
the induced Lie bracket $ [ \BBB(f_1)(-), \BBB(f_2)(-)  ]_c : \BBB(c_1)\otimes\BBB(c_2) \to\BBB(c)$
is the zero map. We have to prove that the induced commutator
$[ \phi_!\,\BBB(f_1)(-), \phi_!\,\BBB(f_2)(-) ]_c : \phi_!\,\BBB(c_1)\otimes
\phi_!\, \BBB(c_2) \to \phi_!\, \BBB(c)$ associated to the functor $(\phi_!)_\ast(\BBB) = \phi_!\, \BBB: 
\CC\to\Alg_{\mathsf{As}}$ is the zero map too. This is an immediate consequence
of our definition of the unital universal enveloping algebra (see \eqref{eqn:UEA} and \eqref{eqn:uUEA})
and the fact that the commutator bracket satisfies the Leibniz rule in both entries. 
The latter property is used to expand the commutator of polynomials to a sum of terms
containing as a factor the commutator of generators, which is identified via \eqref{eqn:UEA} with the Lie bracket.
\end{proof}

As a consequence of Lemma \ref{lem:detection}, we obtain
\begin{cor}\label{cor:pointwisequantization}
The restriction to the categories of field theories
of $(\phi_!)_\ast: {\Alg_{\mathsf{uLie}}}^\CC\to{\Alg_{\mathsf{As}}}^{\CC}$ 
is a model for the left adjoint functor $\Qlin : \LFT(\ovr{\CC})\to\QFT(\ovr{\CC})$ in \eqref{eqn:QUadjunction}.
\end{cor}

\begin{rem}
Let us explain why $\Qlin : \LFT(\ovr{\CC})\to\QFT(\ovr{\CC})$ deserves the name
quantization functor. Suppose that $\BBB = H\, \LLL\in \LFT(\ovr{\CC})$ 
is the composition of a functor $\LLL : \CC\to \mathbf{PSymp}$ to 
the category of presymplectic vector spaces with the Heisenberg Lie algebra functor 
$H : \mathbf{PSymp}\to \Alg_{\mathsf{uLie}}$ as described in Example \ref{ex:LinearFT}.
It is easy to check that the composition $\phi_!\, H  : \mathbf{PSymp}\to \Alg_{\mathsf{As}}$
of the Heisenberg Lie algebra functor and the unital universal enveloping algebra functor
is naturally isomorphic to the usual (polynomial) CCR-algebra functor
$\mathfrak{CCR} :  \mathbf{PSymp}\to \Alg_{\mathsf{As}}$ that is used in
the quantization of linear field theories, see e.g.\ \cite{Baer,BGproc,BeniniDappiaggiHack}.
In particular, we obtain a natural isomorphism
$\Qlin\big(H\,\LLL\big) \cong \mathfrak{CCR}~ \LLL : \CC\to\Alg_{\mathsf{As}}$,
which means that our quantization prescription via $\Qlin$ is in this 
case equivalent to the ordinary CCR-algebra quantization of linear field theories.
\end{rem}

We would like to emphasize that our linear quantization functor
preserves both $j$-locality and $W$-constancy, i.e.\
it preserves descent and the time-slice axiom of field theories.
\begin{cor}\label{cor:quantizationpreserves}
\begin{itemize}
\item[a)] Let $j :\ovr{\CC}\to\ovr{\DD}$ be a full orthogonal subcategory. Then
the linear quantization functor $\Qlin : \LFT(\ovr{\DD})\to \QFT(\ovr{\DD})$
maps $j$-local linear field theories to $j$-local quantum field theories,
(see Definition \ref{def:jlocal}).

\item[b)] Let $\ovr{\CC}$ be an orthogonal category and $W\subseteq \mathrm{Mor}\,\CC$ a subset.
Then the linear quantization functor $\Qlin : \LFT(\ovr{\CC})\to \QFT(\ovr{\CC})$
maps $W$-constant linear field theories to $W$-constant quantum field theories,
(see Definition \ref{def:Wconstant}).
\end{itemize}
\end{cor}
\begin{proof}
Item a) is a consequence of Proposition \ref{prop:preservation}
and item b) is a consequence of Proposition \ref{prop:specialpreservation} and 
Corollary \ref{cor:pointwisequantization}.
\end{proof}

\section{\label{sec:outlook}Towards the quantization of linear gauge theories}
The techniques we developed in this paper can be refined to the case where $\MM$
is a suitable {\em symmetric monoidal model category}. Let us recall that a model category
is a category that comes equipped with three distinguished classes
of morphisms -- called weak equivalences, fibrations and cofibrations --
that satisfy a list of axioms going back to Quillen, see e.g.\ \cite{Dwyer}
for a concise introduction. The main role is played by the weak equivalences,
which introduce a consistent concept of  ``two things being the same'' that is weaker than
the usual concept of categorical isomorphism. For example, the category $\MM = \Ch(\bbK)$
of (possibly unbounded) chain complexes of vector spaces over a field $\bbK$
may be endowed with a symmetric monoidal model category structure in which the
weak equivalences are quasi-isomorphisms, see e.g.\  \cite{Hovey}.
\sk

Model category theory plays an important role in the mathematical formulation
of (quantum) gauge theories. In particular, the `spaces' of fields in a gauge theory
are actually higher spaces called stacks, which may be formalized
within model category theory. We refer to e.g.\ \cite{Schreiber}
for the general framework and also to \cite{BSSStack} for the example of Yang-Mills theory.
Consequently, the observable `algebras' in a quantum gauge theory are actually
higher algebras, e.g.\ the differential graded algebras arising in the BRST/BV formalism.
We refer to \cite{Hollands,FredenhagenRejzner,FredenhagenRejzner2} for concrete
constructions within the BRST/BV formalism in algebraic quantum field theory, to
\cite{BeniniSchenkelWoikehomotopy} for the relevant model categorical perspective
and to \cite{Hawkins} for a related deformation theoretic point of view.
\sk

The aim of this last section is to refine our results for the linear quantization adjunction from 
Section \ref{sec:quantization} to the framework of model category theory.
This will provide a mathematically solid setup to quantize linear gauge theories 
to quantum gauge theories in a way that is consistent with the concept of weak equivalences.
As an explicit example, we discuss the quantization of linear Chern-Simons theory on oriented surfaces.
In order to simplify our analysis, we restrict 
ourselves to the case where $\MM=\Ch(\bbK)$ is the symmetric monoidal model 
category of chain complexes of vector spaces over a field $\bbK$ of characteristic zero, 
e.g.\ $\bbK=\bbC$ or $\bbK=\bbR$. 
In this section we shall freely use terminology and results from general
model category theory \cite{Dwyer,Hovey} and more specifically 
the model structures for colored operads and their algebras  \cite{HinichOriginal,Hinich}.
We refer to \cite{BeniniSchenkelWoikehomotopy,BeniniSchenkelReview} for a more gentle presentation of how these
techniques can be applied to $\Ch(\bbK)$-valued algebraic quantum field theory.

\subsection{Model structures on field theory categories}
Our first (immediate) result is that the categories $\FT(\ovr{\CC},\P^{(r_1,r_2)})$  
of field theories with values in $\MM = \Ch(\bbK)$
from Definition \ref{def:fieldtheory} are model categories, i.e.\ there 
exists a consistent concept of weak equivalences for $\Ch(\bbK)$-valued field theories.
Furthermore, the adjunctions in \eqref{eqn:FTadjunction} are compatible
with these model category structures in the sense that they are {\em Quillen 
adjunctions}.
\begin{propo}\label{prop:modelstructure}
Let $\ovr{\CC}$ be any orthogonal category and $\P^{(r_1,r_2)}\in \Op_{\{\ast\}}^{2\mathrm{pt}}(\Ch(\bbK))$ 
any bipointed single-colored operad. Define a $\FT(\ovr{\CC},\P^{(r_1,r_2)})$-morphism $\zeta : \AAA\to \BBB$
(i.e.\ a natural transformation of functors $\AAA,\BBB: \CC\to \Alg_{\P}$) to be 
\begin{itemize}
\item[(i)] a weak equivalence if the underlying $\Ch(\bbK)$-morphism of
each component $\zeta_c : \AAA(c)\to\BBB(c)$ is a quasi-isomorphism,
\item[(ii)] a fibration if the underlying $\Ch(\bbK)$-morphism of
each component $\zeta_c : \AAA(c)\to\BBB(c)$ is degree-wise surjective,
and 
\item[(iii)] a cofibration if it has the left-lifting property with respect to all acyclic fibrations.
\end{itemize}
These choices define a model structure on $\FT(\ovr{\CC},\P^{(r_1,r_2)})$.
\end{propo}
\begin{proof}
This is a consequence of Theorem \ref{theo:FTcatiso} and Hinich's
results \cite{HinichOriginal,Hinich}, which show that all colored
operads in $\Ch(\bbK)$ are admissible for $\bbK$ a field of characteristic zero.
\end{proof}
\begin{propo}\label{prop:Quillen}
Let $F : \ovr{\CC}\to\ovr{\DD}$ be any orthogonal functor
and $\phi : \P^{(r_1,r_2)}\to\Q^{(s_1,s_2)}$ any $\Op_{\{\ast\}}^{2\mathrm{pt}}(\Ch(\bbK))$-morphism.
Then the adjunction in \eqref{eqn:FTadjunction} is a Quillen adjunction
with respect to the model structures from Proposition \ref{prop:modelstructure}.
\end{propo}
\begin{proof}
This follows immediately from  \cite{HinichOriginal,Hinich}.
\end{proof}

As a specific instance of the general result of Proposition \ref{prop:modelstructure}, 
we obtain that both the category of $\Ch(\bbK)$-valued linear 
field theories $\LFT(\ovr{\CC})$ and the category of $\Ch(\bbK)$-valued
quantum field theories $\QFT(\ovr{\CC})$ carry a canonical model structure.
In order to develop a better intuition for $\Ch(\bbK)$-valued
field theories and their relation to gauge theories, let us introduce
a simple example of a $\Ch(\bbK)$-valued linear field theory.
\begin{ex}[Linear Chern-Simons theory on oriented surfaces]\label{ex:ChernSimons}
Let us denote by $\Man_2$ the category of oriented $2$-manifolds
and orientation preserving open embeddings. We endow
$\Man_2$ with an orthogonality relation $\perp$ by declaring
two morphisms $f_1 : M_1\to M$ and $f_2: M_2\to M$ to be orthogonal,
$f_1\perp f_2$, if and only if their images are disjoint, i.e.\ 
$f_1(M_1)\cap f_2(M_2)=\emptyset$. 
The field configurations on $M\in\Man_2$ of Chern-Simons theory with 
structure group $\bbR$ are given by flat $\bbR$-connections $A\in\Omega^1(M)$
modulo gauge transformations. In the context of linear derived geometry, see e.g.\
\cite{CostelloGwilliam} and \cite{BeniniSchenkelReview} for more details, the
flatness condition $\dd A =0$ and the quotient by gauge 
transformations $A\to A+\dd\epsilon$ are both refined to higher categorical
concepts called homotopy kernels and stacky quotients. This results in a
chain complex of field configurations, that is 
given in our example by the shifted de Rham complex
\begin{flalign}
\mathfrak{F}(M)\,:=\, \Big(
\xymatrix@C=2em{
\stackrel{(-1)}{\Omega^2(M)} 
& \ar[l]_-{\dd} \stackrel{(0)}{\Omega^1(M)} 
& \ar[l]_-{\dd} \stackrel{(1)}{\Omega^0(M)}
}
\Big)\quad,
\end{flalign}
where the round brackets indicate homological degrees.
Note that the zeroth homology of $\mathfrak{F}(M)$ is the usual vector
space of gauge equivalence classes of flat $\bbR$-connections on $M$.
\sk

We describe linear observables on $M$ by the 
smooth dual of $\mathfrak{F}(M)$, which explicitly reads as
\begin{flalign}\label{eqn:LLLcomplex}
\LLL(M)\,:=\, \Big(
\xymatrix@C=2em{
\stackrel{(-1)}{\Omega_\cc^2(M)} 
& \ar[l]_-{-\dd} \stackrel{(0)}{\Omega_\cc^1(M)} 
& \ar[l]_-{-\dd} \stackrel{(1)}{\Omega_\cc^0(M)}
}
\Big)\quad,
\end{flalign}
where the subscript $\cc$ denotes forms with compact support. The evaluation
of observables on field configurations is given by the integration
map $\int_M : \LLL(M)\otimes \mathfrak{F}(M)\to\bbK$. Note that the
minus signs in \eqref{eqn:LLLcomplex} are necessary for
$\int_M$ to be a chain map. To define a presymplectic
structure on $\LLL(M)$, observe that there exists a chain map
\begin{flalign}
\parbox{0.5cm}{\xymatrix{
\LLL(M)\ar[d]_-{\ell}\\
\mathfrak{F}(M)
}
}~~:=~~~~ \left(\parbox{2cm}{\xymatrix@C=2em{
\ar[d]_-{-\iota}\Omega_\cc^2(M) & \ar[d]_-{\iota}\ar[l]_-{-\dd} \Omega_\cc^1(M) & \ar[d]_-{-\iota}
\ar[l]_-{-\dd} \Omega_\cc^0(M)\\
\Omega^2(M) & \ar[l]_-{\dd} \Omega^1(M) & \ar[l]_-{\dd} \Omega^0(M)
}}\right)\quad,
\end{flalign}
where $\iota:\Omega^p_\cc(M)\to\Omega^p(M)$ 
denotes the inclusion of compactly supported $p$-forms into
all $p$-forms. We define the chain map
\begin{flalign}\label{eqn:omegaCS}
\omega ~:~\xymatrix@C=4em{
\LLL(M)\otimes\LLL(M) \ar[r]^-{\id\otimes\ell} & \LLL(M)\otimes\mathfrak{F}(M) \ar[r]^-{\int_M} & \bbK
}
\end{flalign}
and note that $\omega$ is a presymplectic structure, i.e.\ it is graded antisymmetric.
\sk

Precisely as in Example \ref{ex:LinearFT}, we can assign to the presymplectic
chain complex $(\LLL(M),\omega)$ its Heisenberg Lie algebra, 
which we shall denote by $\BBB^{}_{\mathrm{CS}}(M):= \LLL(M)\oplus\bbK \in \Alg_{\mathsf{uLie}}$.
Using pushforwards of compactly supported
forms along $\Man_2$-morphisms $f:M\to N$
and observing that \eqref{eqn:omegaCS} are the components of a natural
transformation, we can promote the assignment $M\mapsto \BBB^{}_{\mathrm{CS}}(M)$ to a functor
$\BBB^{}_{\mathrm{CS}} : \Man_2 \to \Alg_{\mathsf{uLie}}$. Because the integration of any 
product of forms with disjoint support yields zero, this functor
defines a $\Ch(\bbK)$-valued linear field theory $\BBB^{}_{\mathrm{CS}}\in \LFT(\ovr{\Man_2})$
on the orthogonal category $\ovr{\Man_2}$. By construction, this linear field theory
describes linear Chern-Simons theory on oriented surfaces.
\end{ex}

\subsection{\label{subsec:homlinquant}Homotopical properties of linear quantization}
As a specific instance of the general result of Proposition
\ref{prop:Quillen}, we obtain that 
the linear quantization adjunction \eqref{eqn:QUadjunction} is a Quillen adjunction
between the model categories $\LFT(\ovr{\CC})$ and $\QFT(\ovr{\CC})$.
Using the general method of {\em derived functors}
(see e.g.\ \cite{Dwyer,Hovey}), there exists a left derived
linear quantization functor $\bbL\Qlin$ and a right derivation of its right adjoint 
$ \bbR\Ulin$. These two derived functors preserve weak
equivalences and hence they are homotopically meaningful.
The standard procedure to construct derived functors
is to {\em choose} fibrant and cofibrant replacement functors,
denoted by $R: \QFT(\ovr{\CC})\to \QFT(\ovr{\CC})$ and $Q: \LFT(\ovr{\CC})\to \LFT(\ovr{\CC})$,
and to define the right derived functor by
\begin{flalign}\label{eqn:deriveddequantization}
\bbR\Ulin := \Ulin~R \,:\, \QFT(\ovr{\CC})~\longrightarrow~\LFT(\ovr{\CC})\quad
\end{flalign}
and the left derived functor by
\begin{flalign}\label{eqn:derivedquantization}
\bbL \Qlin  := \Qlin~Q\,:\, \LFT(\ovr{\CC})~\longrightarrow ~\QFT(\ovr{\CC}) \quad.
\end{flalign}
Note that there is some flexibility in choosing $R$ and $Q$, but different choices 
define naturally weakly equivalent derived functors.
\sk

For practical applications, it is crucial to find simple models for derived functors
that can be computed explicitly. The goal of this subsection is to obtain
such simple models for the derived functors of the linear quantization 
adjunction \eqref{eqn:QUadjunction}. For the right derived functor $\bbR\Ulin$,
this problem is easy to solve because every object in the model category 
$\QFT(\ovr{\CC})$ is fibrant, hence the identity functor $R=\id$ defines
a fibrant replacement functor. This immediately implies
\begin{propo}
The underived functor $\Ulin : \QFT(\ovr{\CC})\to \LFT(\ovr{\CC})$ is a model
for the right derived functor $\bbR\Ulin$ in \eqref{eqn:deriveddequantization}.
\end{propo}

For the left derived functor $\bbL\Qlin$, i.e.\ the derived linear quantization functor,
the situation gets more complicated because not every object in $\LFT(\ovr{\CC})$
is cofibrant. However, a more detailed study of $\Qlin$ reveals the following
pleasing result.
\begin{propo}\label{prop:Qlinunderived}
The underived  functor $\Qlin : \LFT(\ovr{\CC}) \to \QFT(\ovr{\CC})$
preserves weak equivalences. As a consequence, it is a model for the left derived
functor $\bbL\Qlin$ in \eqref{eqn:derivedquantization}.
\end{propo}
\begin{proof}
Recall from Corollary \ref{cor:pointwisequantization}
that $\Qlin = (\phi_!)_\ast = \phi_!\circ (-)$ is given by post-composing
with the left adjoint functor $\phi_! : \Alg_{\mathsf{uLie}}\to\Alg_{\mathsf{As}}$.
It is shown in Lemma \ref{lem:phi!weakequivalences} that $\phi_!$ preserves weak equivalences,
hence $\Qlin$ preserves weak equivalences as these are defined component-wise 
(see Proposition \ref{prop:modelstructure}).
\sk

To prove the second statement, consider the natural transformation 
\begin{flalign}
\xymatrix@C=3em{
\bbL\Qlin = \Qlin\,Q \ar[r]^-{\Qlin\,q} \,&\, \Qlin
}
\end{flalign}
obtained by whiskering the natural weak equivalence $q : Q\to\id$ corresponding to the
cofibrant replacement functor $Q$. This is a natural weak equivalence
because $\Qlin$ preserves weak equivalences, hence
$\Qlin$ is a model for $\bbL\Qlin$.
\end{proof}

\begin{ex}\label{ex:ChernSimonsQuantization}
Let us apply these results to carry out the quantization of linear Chern-Simons theory
from Example \ref{ex:ChernSimons}. As shown in Proposition \ref{prop:Qlinunderived},
the {\em underived} linear quantization functor $\Qlin$ provides a homotopically 
meaningful quantization prescription that agrees (up to weak equivalence) with any derived
linear quantization functor $\bbL\Qlin$. Applying the functor $\Qlin$ to the
linear Chern-Simons model $\BBB_{\mathrm{CS}}^{}\in \LFT(\ovr{\Man_2})$
from Example \ref{ex:ChernSimons}, we obtain the linear Chern-Simons 
quantum field theory $\AAA_{\mathrm{CS}}^{} := \Qlin(\BBB_{\mathrm{CS}}^{})\in \QFT(\ovr{\Man_2})$.
Using the concrete description of $\Qlin$ given in Section \ref{sec:quantization},
we can compute explicitly the differential graded algebra $\AAA_{\mathrm{CS}}^{}(M)\in\Alg_{\mathsf{As}}$
that is assigned to an oriented $2$-manifold $M\in\Man_2$. One observes 
that this differential graded algebra is presentable by generators and relations.
The generators \eqref{eqn:LLLcomplex} are concentrated in homological degrees $-1$, $0$ and $1$.
We shall use the intuitive `smeared quantum field' notation to denote 
the generators by
\begin{subequations}
\begin{flalign}
\widehat{C}(\chi) \in \AAA_{\mathrm{CS}}^{}(M)_{-1}\quad,\qquad
\widehat{A}(\alpha) \in \AAA_{\mathrm{CS}}^{}(M)_0\quad,\qquad
\widehat{B}(\beta) \in \AAA_{\mathrm{CS}}^{}(M)_1\quad,
\end{flalign}
for all $\chi\in\Omega^2_\cc(M)$, $\alpha\in\Omega^1_\cc(M)$ and $\beta\in\Omega^0_\cc(M)$.
The differential acts on these generators as
\begin{flalign}
\dd \widehat{B}(\beta) \,=\, \widehat{A}(-\dd\beta)\quad,\qquad
\dd \widehat{A}(\alpha) \,=\, \widehat{C}(-\dd \alpha)\quad,\qquad
\dd \widehat{C}(\chi) \,=\,0\quad.
\end{flalign}
\end{subequations}
The relations are as follows:
\begin{itemize}
\item {\em $\bbK$-linearity:} $\widehat{C}(k\,\chi+ k^\prime\,\chi^\prime) = 
k\,\widehat{C}(\chi) + k^\prime\,\widehat{C}(\chi^\prime)$, for all $\chi,\chi^\prime\in \Omega^2_\cc(M)$
and $k,k^\prime\in \bbK$, and similarly for $\widehat{A}$ and $\widehat{B}$;

\item {\em Commutation relations:} The non-vanishing graded commutators are
\begin{subequations}
\begin{flalign}
\big[\widehat{A}(\alpha),\widehat{A}(\alpha^\prime)\big]\,&=\,\omega(\alpha,\alpha^\prime) \,=\,\int_M \alpha\wedge\alpha^\prime\quad,\\
\big[\widehat{C}(\chi),\widehat{B}(\beta)\big] \,&=\, \omega(\chi,\beta) \,=\, - \int_M\chi\wedge\beta\quad,\\
\big[\widehat{B}(\beta),\widehat{C}(\chi)\big] \,&=\, \omega(\beta,\chi) \,=\, - \int_M\beta\wedge\chi\quad.
\end{flalign}
\end{subequations}
\end{itemize}
Note that the zeroth homology $H_0(\AAA_{\mathrm{CS}}(M))$ is the ordinary
algebra of gauge invariant observables for quantized flat $\bbR$-connections,
see e.g.\ \cite{DappiaggiMurroSchenkel}.
\end{ex}

\subsection{Homotopy $j$-locality and homotopy $W$-constancy}
We would like to conclude by introducing natural homotopical generalizations
of the $j$-locality property (see Definition \ref{def:jlocal}) and the $W$-constancy property
(see Definition \ref{def:Wconstant}) in the context of model category theory. 
It will be shown that these properties are preserved by linear quantization.

\paragraph{Homotopy $j$-locality:}
Let $j : \ovr{\CC}\to\ovr{\DD}$ be a full orthogonal subcategory
and $\P^{(r_1,r_2)}$ a bipointed single-colored operad.
From Proposition \ref{prop:Quillen}, we obtain a Quillen adjunction
$j_! : \FT(\ovr{\CC},\P^{(r_1,r_2)}) \rightleftarrows \FT(\ovr{\DD},\P^{(r_1,r_2)}) :j^\ast$.
For the right derived functor we can choose again the underived functor
$\bbR j^\ast := j^\ast$, because every object in $\FT(\ovr{\DD},\P^{(r_1,r_2)})$ is fibrant.
However, in contrast to the linear quantization functor from the previous subsection,
the left adjoint functor $j_!$ in general does {\em not} preserve weak equivalences
and hence it has to be derived $\bbL j_!:= j_!\,Q$. (See \cite[Appendix A]{BeniniSchenkelWoikehomotopy}
for concrete examples illustrating this fact.) As a consequence, our previous concept 
of $j$-locality from Definition \ref{def:jlocal} has to be derived as well
in order to be homotopically meaningful. In what follows we denote by
$q : Q\to \id$ the natural weak equivalence corresponding to our 
choice of cofibrant replacement functor $Q$.
\begin{defi}\label{def:homotopyjlocal}
A field theory $\AAA \in \FT\big(\ovr{\DD},\P^{(r_1,r_2)}\big)$ 
is called {\em homotopy $j$-local} if the corresponding
component  of the derived counit
\begin{flalign}\label{eqn:homotopyjlocal}
\widetilde{\epsilon}_{\AAA}^{}~:~\xymatrix@C=4em{
j_!\,Q\,j^\ast(\AAA) \ar[r]^-{j_! q_{j^\ast(\AAA)}^{}} & j_!\,j^\ast(\AAA)\ar[r]^-{\epsilon_\AAA^{}} & \AAA
}
\end{flalign}
is a weak equivalence in $\FT\big(\ovr{\DD},\P^{(r_1,r_2)}\big)$.
\end{defi}

\begin{propo}\label{propo:quantizationpreserveshomotopyjlocal}
The linear quantization functor $\Qlin : \LFT(\ovr{\DD}) \to \QFT(\ovr{\DD})$ (see Proposition
\ref{prop:Qlinunderived}) maps homotopy $j$-local linear field theories 
to homotopy $j$-local quantum field theories.
\end{propo}
\begin{proof}
Let $\BBB\in\LFT(\ovr{\DD})$  be a homotopy $j$-local linear field theory,
i.e.\ $\widetilde{\epsilon}_{\BBB}^{} : j_! \,Q\, j^\ast(\BBB)\to \BBB$ is a weak equivalence.
We have to prove that the derived counit 
$\widetilde{\epsilon}_{\Qlin(\BBB)}^{} : j_!\,Q\,j^\ast\,\Qlin(\BBB) \to \Qlin(\BBB)$ 
corresponding to the quantum field theory $\Qlin(\BBB)\in \QFT(\ovr{\DD})$
is a weak equivalence too. 
\sk

As a preparatory step, let us consider the commutative diagram
\begin{flalign}
\xymatrix@C=8em{
j_!\,Q\,j^\ast\,\Qlin(\BBB)  \ar[r]^-{\widetilde{\epsilon}_{\Qlin(\BBB)}^{}} & \Qlin(\BBB)\\
\ar[u]^-{j_!\,Q\,j^\ast\,\Qlin \,\widetilde{\epsilon}_\BBB^{}}_-{\sim} 
j_!\,Q\,j^\ast\,\underbrace{\Qlin \, j_!}_{\cong \,j_!\,\Qlin}\,Q\,j^\ast (\BBB) \ar[r]^-{\widetilde{\epsilon}_{\Qlin \, j_!\,Q\,j^\ast(\BBB)}^{}} & 
\ar[u]_-{\Qlin \,\widetilde{\epsilon}_\BBB^{}}^-{\sim} \underbrace{\Qlin\, j_!}_{\cong\,j_!\,\Qlin}\,Q\,j^\ast (\BBB)\\
\ar[u]^-{j_!\,Q\,j^\ast\, j_!\,q_{\Qlin\,Q\,j^\ast (\BBB)}^{}}_-{\sim} j_!\,Q\,j^\ast\, j_!\,Q\,\Qlin\,Q\,j^\ast (\BBB)\ar[r]_-{\widetilde{\epsilon}_{ j_!\,Q\,\Qlin\,Q\,j^\ast(\BBB)}^{}}
& j_!\,Q\,\Qlin\,Q\,j^\ast (\BBB)\ar[u]_-{ j_!\, q_{\Qlin\,Q\,j^\ast (\BBB)}^{}}^-{\sim}
}
\end{flalign}
The vertical arrows in the top square are weak equivalences because $\BBB$ is 
by hypothesis homotopy $j$-local. The vertical arrows in the bottom square
are weak equivalences because left Quillen functors preserve cofibrant objects and weak equivalences
between cofibrant objects. The natural isomorphism in the underbraces is due to
\eqref{eqn:generalsquareadjunctions}. By the $2$-of-$3$ property of
weak equivalences, the top horizontal arrow is a weak equivalence (which is our claim)
if and only if the bottom horizontal arrow is a weak equivalence.
\sk

Introducing $\AAA := \Qlin\,Q\,j^\ast(\BBB)$, it thus
remains to prove that $\widetilde{\epsilon}_{j_!\,Q(\AAA)} : j_! \,Q\,j^\ast\,j_! \,Q(\AAA)\to j_!\,Q(\AAA)$
is a weak equivalence. This follows from the $2$-of-$3$ property of
weak equivalences and the commutative diagram
\begin{flalign}
\xymatrix@C=5em{
\ar@/^2.5pc/[rr]^-{\widetilde{\epsilon}_{j_!\,Q(\AAA)}}  j_!\,Q\,j^\ast\,j_!\,Q(\AAA) \ar[r]^-{j_! q_{j^\ast j_! Q(\AAA)}^{}} & j_!\,j^\ast\,j_!\,Q(\AAA) \ar[r]^-{\epsilon_{j_! Q(\AAA)}^{}} & j_! \,Q(\AAA)\\
\ar[u]^-{j_! Q \eta_{Q(\AAA)}^{}}_-{\cong} j_! \,Q^2(\AAA) \ar[r]_-{j_! q^{}_{Q(\AAA)}}^-{\sim} & \ar[u]^-{j_! \eta_{Q(\AAA)}^{}}_-{\cong} j_! \, Q(\AAA) \ar@{=}[ur]&
}
\end{flalign}
where the top triangle is just the definition of the derived counit \eqref{eqn:homotopyjlocal}.
The vertical arrows are isomorphisms because of Proposition \ref{prop:extension}
and the right triangle is the triangle identity for the (non-derived) unit and counit.
The bottom horizontal arrow is a weak equivalence because left Quillen functors 
preserve weak equivalences between cofibrant objects.
\end{proof}

\begin{ex}
Let $j: \ovr{\Manc_2}\to \ovr{\Man_2}$ be the full orthogonal subcategory
of all oriented $2$-manifolds $M$ that are diffeomorphic to $\bbR^2$.
It is an interesting question whether the linear Chern-Simons
quantum field theory $\AAA_{\mathrm{CS}}^{}\in\QFT(\ovr{\Man_2})$ from Example
\ref{ex:ChernSimonsQuantization} is homotopy $j$-local with respect to this $j$.
In particular, homotopy $j$-locality would imply that its value
$\AAA_{\mathrm{CS}}^{}(M)$ on a topologically non-trivial oriented $2$-manifold $M$ 
such as the torus is already encoded in the restriction 
$j^\ast(\AAA_{\mathrm{CS}}^{})\in\QFT(\ovr{\Manc_2})$ of the quantum field theory to disks.
Unfortunately, proving homotopy $j$-locality of a given theory is a complicated task
and hence we can not yet provide an answer to the question whether $\AAA_{\mathrm{CS}}^{}\in\QFT(\ovr{\Man_2})$
is homotopy $j$-local or not. We however would like to mention that 
positive results are already available for simple toy-models
which do not involve quantization, see \cite{BeniniSchenkelWoikehomotopy} for details.
We expect that Proposition \ref{propo:quantizationpreserveshomotopyjlocal}
will be very useful for investigating homotopy $j$-locality of 
$\AAA_{\mathrm{CS}}^{}=\Qlin(\BBB_{\mathrm{CS}}^{})\in\QFT(\ovr{\Man_2})$
because it allows us to replace this question by the (slightly) simpler question
whether the linear field theory $\BBB_{\mathrm{CS}}^{}\in \LFT(\ovr{\Man_2})$ from Example
\ref{ex:ChernSimons} is homotopy $j$-local. We hope to come back to this issue
in a future work.
\end{ex}

\paragraph{Homotopy $W$-constancy:}
Let $\ovr{\CC}$ be an orthogonal category, $W\subseteq \mathrm{Mor}\,\CC$ a subset
and $\P^{(r_1,r_2)}$ a bipointed single-colored operad.
Similarly to locally constant factorization algebras \cite{CostelloGwilliam},
we propose a homotopical generalization of the $W$-constancy property from
Definition \ref{def:Wconstant}.
\begin{defi}\label{def:homotopyWconstant}
A field theory $\AAA\in \FT\big(\ovr{\CC},\P^{(r_1,r_2)}\big)$ is called 
{\em homotopy $W$-constant} if the $\Ch(\bbK)$-morphism underlying
the $\Alg_\P$-morphism $\AAA(f) : \AAA(c)\to \AAA(c^\prime)$ is a 
quasi-isomorphism for all $(f:c\to c^\prime)\in W$.
\end{defi}

\begin{propo}\label{propo:quantizationpreserveshomotopyWconstant}
The linear quantization functor $\Qlin : \LFT(\ovr{\CC}) \to \QFT(\ovr{\CC})$ (see Proposition
\ref{prop:Qlinunderived}) maps homotopy $W$-constant linear field theories 
to homotopy $W$-constant quantum field theories.
\end{propo}
\begin{proof}
Recall from Corollary \ref{cor:pointwisequantization}
that $\Qlin = (\phi_!)_\ast = \phi_!\circ (-)$ is given by post-composing
with the left adjoint functor $\phi_! : \Alg_{\mathsf{uLie}}\to\Alg_{\mathsf{As}}$.
By Lemma \ref{lem:phi!weakequivalences}, the latter preserves weak equivalences
and hence it preserves the homotopy $W$-constancy property.
\end{proof}

\begin{ex}
It is easy to prove that the linear Chern-Simons model $\BBB_{\mathrm{CS}}^{}\in\LFT(\ovr{\Man_2})$ 
from Example \ref{ex:ChernSimons} is homotopy $W$-constant for $W\subseteq\mathrm{Mor}\, \Man_{2}$
the subset of all isotopy equivalences $f:M\to M^\prime$. The chain map underlying
$\BBB_{\mathrm{CS}}^{}(f) : \BBB_{\mathrm{CS}}^{}(M)\to \BBB_{\mathrm{CS}}^{}(M^\prime)$ 
is given by
\begin{flalign}
f_\ast\oplus \id \,:\, \LLL(M)\oplus\bbK ~\longrightarrow ~\LLL(M^\prime)\oplus \bbK\quad,
\end{flalign}
where $\LLL(M)$ and $\LLL(M^{\prime})$ are (up to a global minus sign) 
shifted compactly supported de Rham complexes (see \eqref{eqn:LLLcomplex})
and $f_\ast$ is given by degree-wise pushforward of compactly supported forms. 
For $f:M\to M^\prime$ an isotopy equivalence, this map is a quasi-isomorphism
because compactly supported de Rham cohomology is invariant under isotopies.
Together with Proposition \ref{propo:quantizationpreserveshomotopyWconstant},
this implies that the linear Chern-Simons quantum field theory
$\AAA_{\mathrm{CS}}^{} = \Qlin(\BBB_{\mathrm{CS}}^{})\in\QFT(\ovr{\Man_2})$
from Example \ref{ex:ChernSimonsQuantization} is homotopy $W$-constant too.
\end{ex}


\section*{Acknowledgments}
We would like to thank the anonymous referee 
for useful comments that helped us to improve this manuscript.
We also would like to thank Marco Benini and Lukas Woike for useful discussions.
S.B.\ is supported by a PhD scholarship of the Royal Society (UK).
A.S.\ gratefully acknowledges the financial support of 
the Royal Society (UK) through a Royal Society University 
Research Fellowship, a Research Grant and an Enhancement Award.


\appendix

\section{\label{app:envelop}Technical details for Section \ref{subsec:homlinquant}}
In this appendix we let $\MM = \Ch(\bbK)$
be the symmetric monoidal model category of chain complexes
of vector spaces over a field $\bbK$ of characteristic zero.
Recall the unital universal enveloping algebra functor 
$\phi_! : \Alg_{\mathsf{uLie}}\to \Alg_{\mathsf{As}}$ 
from Lemma \ref{lem:phi!LieAss}. With Proposition \ref{prop:modelstructure} we see that
the categories $ \Alg_{\mathsf{uLie}}$ and $\Alg_{\mathsf{As}}$
carry a canonical model structure in which a morphism
is a weak equivalence if the underlying $\Ch(\bbK)$-morphism is a quasi-isomorphism.
\begin{lem}\label{lem:phi!weakequivalences}
The functor $\phi_! : \Alg_{\mathsf{uLie}}\to \Alg_{\mathsf{As}}$ preserves weak equivalences.
\end{lem}
\begin{proof}
As a preparation for the proof, we have to revisit the unital universal
enveloping algebra construction from \eqref{eqn:UEA} and \eqref{eqn:uUEA}
for the category $ \Ch(\bbK)$ of chain complexes. By definition, 
the unital universal enveloping algebra $\phi_!(V)\in \Alg_{\mathsf{As}}$
of a unital Lie algebra $(V,[\cdot,\cdot],\eta)\in \Alg_{\mathsf{uLie}}$ is the differential 
graded algebra presented by
\begin{flalign}
\phi_!(V)\,=\, T^\otimes V\big/ \mathcal{I}\quad,
\end{flalign}
where $\mathcal{I}\subseteq T^\otimes V$ is the differential graded ideal
generated by the relations 
\begin{flalign}\label{eqn:explicitrelationsTMP}
v_1\otimes v_2 - (-1)^{\vert v_1\vert\,\vert v_2\vert}~v_2\otimes v_1 \,=\, [v_1,v_2]\quad,\qquad
 \oone_\otimes \,=\, \oone\quad,
\end{flalign}
for all homogeneous elements $v_1,v_2\in V$. Here $\oone_\otimes = \eta_\otimes(1)$
denotes the unit element of the tensor algebra $T^\otimes V = \bigoplus_{m=0}^{\infty} V^{\otimes m}$
and $\oone = \eta(1)$ is the unit element of $V$. Using the canonical filtration
$T^{\leq n} V := \bigoplus_{m= 0}^n V^{\otimes m}$
of the tensor algebra, we define
\begin{subequations}\label{eqn:sequentialcolimitTMP}
\begin{flalign}
\phi_!(V)^{n}\, :=\, T^{\leq n} V\big/ \big( T^{\leq n} V \cap \mathcal{I}\big)\quad,
\end{flalign}
for all $n\geq 0$. This defines a sequential diagram
\begin{flalign}
\xymatrix{
\phi_!(V)^{0}\ar@{^{(}->}[r]~&~\phi_!(V)^{1} \ar@{^{(}->}[r]~&~ \phi_!(V)^{2}  \ar@{^{(}->}[r] ~&~ \cdots
}
\end{flalign}
in the category $\Ch(\bbK)$, whose colimit 
\begin{flalign}
\phi_!(V) \,=\, \colim_{n\geq 0} \big(\phi_!(V)^{n}\big)
\end{flalign}
\end{subequations}
is the chain complex underlying $\phi_!(V)$.
Using the explicit form of the relations \eqref{eqn:explicitrelationsTMP},
we observe that the quotient
\begin{flalign}
\phi_!(V)^{n+1}\big/ \phi_!(V)^{n}~\cong~\widetilde{V}^{\otimes n+1}\big/ \Sigma_{n+1}
\end{flalign}
is given by the coinvariants of the canonical permutation group action
on the $n+1$-fold tensor product of quotient chain complex $\widetilde{V} := V/\bbK\oone$.
In other words, there is a short exact sequence
\begin{flalign}\label{eqn:SESTMP}
\xymatrix{
0 \ar[r] ~&~ \phi_!(V)^{n} \ar[r] ~&~ \phi_!(V)^{n+1} \ar[r] ~&~ \widetilde{V}^{\otimes n+1}\big/ \Sigma_{n+1}\ar[r] ~&~ 0
}
\end{flalign}
of chain complexes, for all $n\geq 0$.
\sk

Let now $\rho : V\to V^\prime$ be a weak equivalence in $\Alg_{\mathsf{uLie}}$. Our goal is to prove
that $\phi_!(\rho) :  \phi_!(V)\to \phi_!(V^\prime)$ is a weak equivalence in $\Alg_{\mathsf{As}}$, i.e.\
that the induced map $H_\bullet(\phi_!(\rho)) : H_\bullet(\phi_!(V))\to H_\bullet(\phi_!(V^\prime)) $
in homology is an isomorphism. Because filtered colimits commute with forming homologies,
it is by \eqref{eqn:sequentialcolimitTMP} sufficient to prove that
\begin{flalign}\label{eqn:homologynisoTMP}
H_\bullet(\phi_!(\rho)^n) \,:\, H_\bullet(\phi_!(V)^n)~\longrightarrow~ H_\bullet(\phi_!(V^\prime)^n) 
\end{flalign}
is an isomorphism, for all $n\geq 0$. For $n=0$ this is clearly the case,
and for $n=1$ it follows from the observation that the induced map
$\rho : \widetilde{V}= V/\bbK\oone\to\widetilde{V}^\prime=V^\prime/\bbK\oone^\prime$
is a quasi-isomorphism. 
\sk

The case $n> 1$ is proven by induction. Assume that 
\eqref{eqn:homologynisoTMP} is an isomorphism for some $n\geq 0$.
The short exact sequence \eqref{eqn:SESTMP} of chain complexes 
induces a long exact sequence in homology and
$\rho : V\to V^\prime$ induces a map of long exact sequences
\begin{flalign}\label{eqn:TMPbigsquare}
\xymatrix@C=1.2em{
\ar[d]H_{m+1}(\tfrac{\widetilde{V}^{\otimes n+1}}{\Sigma_{n+1}})\ar[r]&
\ar[d]^-{\cong}H_m( \phi_!(V)^{n}) \ar[r]& 
\ar[d]H_m(\phi_!(V)^{n+1}) \ar[r]& 
\ar[d]H_{m}(\tfrac{\widetilde{V}^{\otimes n+1}}{\Sigma_{n+1}}) \ar[r]&
\ar[d]^-{\cong}H_{m-1}( \phi_!(V)^{n})\\
H_{m+1}(\tfrac{\widetilde{V}^{\prime \otimes n+1}}{\Sigma_{n+1}})\ar[r]&
H_m( \phi_!(V^\prime)^{n}) \ar[r]& 
H_m(\phi_!(V^\prime)^{n+1}) \ar[r]& 
H_{m}(\tfrac{\widetilde{V}^{\prime\otimes n+1}}{\Sigma_{n+1}}) \ar[r]&
H_{m-1}( \phi_!(V^\prime)^{n})
}
\end{flalign}
where by our induction hypothesis the second and fifth vertical maps are isomorphisms.
The induction step consists of proving that the middle vertical map is an isomorphism. By the five lemma,
this follows if the first and fourth vertical maps are isomorphisms.
Because $\bbK$ is
a field of characteristic zero, the chain complex
$\widetilde{V}^{\otimes n+1}/\Sigma_{n+1}$ of coinvariants 
of the permutation group action is isomorphic to the chain complex 
$(\widetilde{V}^{\otimes n+1})^{\Sigma_{n+1}}$ of invariants of the permutation group action.
For $\mathrm{char}\,\bbK=0$, taking invariants of finite group actions preserves quasi-isomorphisms, 
so the K\"unneth theorem 
and the fact that $\rho : \widetilde{V}\to\widetilde{V}^\prime$ is a quasi-isomorphism
imply that the first and fourth vertical maps in \eqref{eqn:TMPbigsquare} are isomorphisms. 
This completes the proof.
\end{proof}



\begin{thebibliography}{10}

\bibitem[BGP07]{Baer}
C.~B\"ar, N.~Ginoux and F.~Pf\"affle, 
{\it Wave equations on Lorentzian manifolds and quantization},
Z\"urich, Switzerland: Eur.\ Math.\ Soc.\ (2007) 
[arXiv:0806.1036 [math.DG]].


\bibitem[BG11]{BGproc}
C.~B\"ar and N.~Ginoux, 
``Classical and Quantum Fields on Lorentzian Manifolds,'' 
Springer Proc.\ Math.\ {\bf 17}, 359 (2011) 
[arXiv:1104.1158 [math-ph]].


\bibitem[BDH13]{BeniniDappiaggiHack} 
M.~Benini, C.~Dappiaggi and T.~P.~Hack,
``Quantum Field Theory on Curved Backgrounds -- A Primer,''
Int.\ J.\ Mod.\ Phys.\ A {\bf 28}, 1330023 (2013)
[arXiv:1306.0527 [gr-qc]].


\bibitem[BDS18]{BeniniDappiaggiSchenkel}
M.~Benini, C.~Dappiaggi and A.~Schenkel,
``Algebraic quantum field theory on spacetimes with timelike boundary,''
Annales Henri Poincar{\'e} {\bf 19}, no.\ 8, 2401 (2018)
[arXiv:1712.06686 [math-ph]].
  

\bibitem[BS17]{BeniniSchenkelPoisson} 
M.~Benini and A.~Schenkel,
``Poisson algebras for non-linear field theories in the Cahiers topos,''
Annales Henri Poincar{\'e} {\bf 18}, no.\ 4, 1435 (2017)
[arXiv:1602.00708 [math-ph]].


\bibitem[BS19]{BeniniSchenkelReview}
M.~Benini and A.~Schenkel,
``Higher Structures in Algebraic Quantum Field Theory,''
{\em to appear in Fortschritte der Physik}
[arXiv:1903.02878 [hep-th]].


\bibitem[BSS18]{BSSStack} 
M.~Benini, A.~Schenkel and U.~Schreiber,
``The stack of Yang-Mills fields on Lorentzian manifolds,''
Commun.\ Math.\ Phys.\  {\bf 359}, no.\ 2, 765 (2018)
[arXiv:1704.01378 [math-ph]].


\bibitem[BSW17]{BeniniSchenkelWoike} 
M.~Benini, A.~Schenkel and L.~Woike,
``Operads for algebraic quantum field theory,''
arXiv:1709.08657 [math-ph].


\bibitem[BSW19]{BeniniSchenkelWoikehomotopy} 
M.~Benini, A.~Schenkel and L.~Woike,
``Homotopy theory of algebraic quantum field theories,''
Lett.\ Math.\ Phys.\  {\bf 109}, no.\ 7, 1487 (2019)
[arXiv:1805.08795 [math-ph]].


\bibitem[BM07]{BergerMoerdijk}
C.~Berger and I.~Moerdijk,
``Resolution of coloured operads and rectification of homotopy algebras,''
in: A.~Davydov, M.~Batanin, M.~Johnson, S.~Lack and A.~Neeman (eds.),
{\it Categories in algebra, geometry and mathematical physics},
Contemp.\ Math.\ {\bf 431}, 31--58,  
American Mathematical Society, Providence, RI (2007). 


\bibitem[BDFY15]{AQFTbook}
R.~Brunetti, C.~Dappiaggi, K.~Fredenhagen and J.~Yngvason,
{\it Advances in algebraic quantum field theory},
Springer Verlag, Heidelberg (2015).


\bibitem[BFR12]{BrunettiFredenhagenRibeiro} 
R.~Brunetti, K.~Fredenhagen and P.~Lauridsen Ribeiro,
``Algebraic Structure of Classical Field Theory I: Kinematics and Linearized Dynamics for Real Scalar Fields,''
arXiv:1209.2148 [math-ph].


\bibitem[BFV03]{Brunetti} 
R.~Brunetti, K.~Fredenhagen and R.~Verch,
``The generally covariant locality principle: A new paradigm for local quantum field theory,''
Commun.\ Math.\ Phys.\  {\bf 237}, 31 (2003)
[math-ph/0112041].


\bibitem[Col16]{Collini}
G.~Collini, 
{\it Fedosov quantization and perturbative quantum field theory},
PhD thesis, Universit\"at Leipzig (2016) [arXiv:1603.09626 [math-ph]]. 


\bibitem[CG17]{CostelloGwilliam}
K.~Costello and O.~Gwilliam,
{\it Factorization algebras in quantum field theory. Vol.\ 1},
New Mathematical Monographs {\bf 31}, 
Cambridge University Press, Cambridge (2017).


\bibitem[DMS17]{DappiaggiMurroSchenkel}
C.~Dappiaggi, S.~Murro and A.~Schenkel,
``Non-existence of natural states for Abelian Chern–Simons theory,''
J.\ Geom.\ Phys.\  {\bf 116}, 119 (2017)
[arXiv:1612.04080 [math-ph]].
  
  
\bibitem[DS95]{Dwyer}
W.~G.~Dwyer and J.~Spalinski,
``Homotopy theories and model categories,''
in: I.~M.~James (ed.), {\it Handbook of Algebraic Topology}, 
73, North-Holland, Amsterdam (1995). 


\bibitem[FV15]{FewsterVerch}
C.~J.~Fewster and R.~Verch,
``Algebraic quantum field theory in curved spacetimes,''
in: R.~Brunetti, C.~Dappiaggi, K.~Fredenhagen and J.~Yngvason (eds.),
{\it Advances in algebraic quantum field theory}, 125--189,
Springer Verlag, Heidelberg (2015)
[arXiv:1504.00586 [math-ph]].


\bibitem[Fre90]{Fre1}
K.~Fredenhagen, 
``Generalizations of the theory of superselection sectors,''
in: D.~Kastler (ed.),
{\em The algebraic theory of superselection sectors: Introduction and recent results}, 
World Scientific Publishing, 379 (1990).


\bibitem[Fre93]{Fre2}
K.~Fredenhagen, 
``Global observables in local quantum physics,''
in: H.~Araki, K.~R.~Ito, A.~Kishimoto and I.~Ojima (eds.),
{\em Quantum and non-commutative analysis: Past, present and future perspectives}, 41--51,
Kluwer Academic Publishers (1993).
  

\bibitem[FRS92]{Fre3}
K.~Fredenhagen, K.-H.~Rehren and B.~Schroer, 
``Superselection sectors with braid group statistics and exchange algebras II: 
Geometric aspects and conformal covariance,'' 
Rev.\ Math.\ Phys.\ {\bf 4}, 113 (1992).


\bibitem[FR12]{FredenhagenRejzner} 
K.~Fredenhagen and K.~Rejzner,
``Batalin-Vilkovisky formalism in the functional approach to classical field theory,''
Commun.\ Math.\ Phys.\  {\bf 314}, 93 (2012)
[arXiv:1101.5112 [math-ph]].
  
  
\bibitem[FR13]{FredenhagenRejzner2} 
K.~Fredenhagen and K.~Rejzner,
``Batalin-Vilkovisky formalism in perturbative algebraic quantum field theory,''
Commun.\ Math.\ Phys.\  {\bf 317}, 697 (2013)
[arXiv:1110.5232 [math-ph]].


\bibitem[GH18]{GwilliamHaugseng} 
O.~Gwilliam and R.~Haugseng,
``Linear Batalin-Vilkovisky quantization as a functor of $\infty$-categories,''
Selecta Mathematica {\bf 24}, 1247  (2018)
[arXiv:1608.01290 [math.AT]].


\bibitem[HK64]{HaagKastler}
R.~Haag and D.~Kastler,
``An algebraic approach to quantum field theory,''
J.\ Math.\ Phys.\  {\bf 5}, 848 (1964).


\bibitem[Haw18]{Hawkins}
E.~Hawkins,
``A Cohomological Perspective on Algebraic Quantum Field Theory,''
Commun.\ Math.\ Phys.\  {\bf 360}, no.\ 1, 439 (2018)
[arXiv:1612.05161 [math-ph]].
  
  
\bibitem[Hin97]{HinichOriginal}
V.~Hinich, 
``Homological algebra of homotopy algebras,'' 
Comm.\ Algebra {\bf 25}, no.\ 10, 3291--3323  (1997)
[arXiv:q-alg/9702015].
Erratum: arXiv:math/0309453 [math.QA].


\bibitem[Hin15]{Hinich}
V.~Hinich, 
``Rectification of algebras and modules,'' 
Doc.\ Math.\ {\bf 20}, 879--926 (2015)
[arXiv:1311.4130 [math.QA]].


\bibitem[Hol08]{Hollands} 
S.~Hollands,
``Renormalized Quantum Yang-Mills Fields in Curved Spacetime,''
Rev.\ Math.\ Phys.\  {\bf 20}, 1033 (2008)
[arXiv:0705.3340 [gr-qc]].


\bibitem[Hov99]{Hovey}
M.~Hovey,
{\it Model categories}, 
Math.\ Surveys Monogr.\ {\bf 63}, 
Amer.\ Math.\ Soc., Providence, RI (1999).


\bibitem[Lan14]{Lang}
B.~Lang,
{\it Universal constructions in algebraic and locally covariant quantum field theory,} 
PhD thesis, University of York (2014).


\bibitem[RV12]{Ruzzi} 
G.~Ruzzi and E.~Vasselli,
``A new light on nets of $C^\ast$-algebras and their representations,''
Commun.\ Math.\ Phys.\  {\bf 312}, 655 (2012)
[arXiv:1005.3178 [math.OA]].
  

\bibitem[Sch13]{Schreiber}
U.~Schreiber,
``Differential cohomology in a cohesive infinity-topos,''
current version available at \url{https://ncatlab.org/schreiber/show/differential+cohomology+in+a+cohesive+topos}
[arXiv:1310.7930 [math-ph]].


\bibitem[Yau16]{Yau}
D.~Yau,
{\it Colored operads},
Graduate Studies in Mathematics {\bf 170},
American Mathematical Society, Providence, RI (2016).


\end{thebibliography}
\end{document}